\documentclass[runningheads]{llncs}
\usepackage{pgfplots}
\usepackage{wrapfig}
\pgfplotsset{compat=1.17}
\usepackage{graphicx}
\usepackage{amsmath}
\usepackage{amssymb}
\usepackage{stmaryrd}
\usepackage{enumitem}
\usepackage{bussproofs}
\EnableBpAbbreviations
\usepackage{mathpartir}
\usepackage{xspace}
\usepackage{bm}
\usepackage{xcolor}

\usepackage{prettyref}
\usepackage[colorlinks=true,citecolor=black,linkcolor=vblue,urlcolor=vblue]{hyperref}
\newrefformat{sec}{Section\,\ref{#1}}
\newrefformat{def}{Def.\,\ref{#1}}
\newrefformat{thm}{Theorem\,\ref{#1}}
\newrefformat{prop}{Proposition\,\ref{#1}}
\newrefformat{lem}{Lemma\,\ref{#1}}
\newrefformat{cor}{Corollary\,\ref{#1}}
\newrefformat{rem}{Remark\,\ref{#1}}
\newrefformat{ex}{Example\,\ref{#1}}
\newrefformat{tab}{Table\,\ref{#1}}
\newrefformat{fig}{Fig.\,\ref{#1}}
\newrefformat{app}{Appendix\,\ref{#1}}

\usepackage{math}
\usepackage[simplenames,differentialdL]{dL}
\usepackage{commands}

\makeatletter
\renewcommand*\l@author[2]{}
\renewcommand*\l@title[2]{}
\makeatother
\usepackage{ifpdf}
\ifpdf
\fi

\title{A Real-Analytic Approach to Differential-Algebraic Dynamic Logic}
\author{Jonathan Hellwig\orcidID{0009-0009-5530-3256} \and
  Andr\'{e} Platzer\orcidID{0000-0001-7238-5710}
}
\institute{
  Karlsruhe Institute of Technology, Karlsruhe, Germany\\
\email{\{jonathan.hellwig, platzer\}@kit.edu}}

\begin{document}
\maketitle
\begin{abstract}
  This paper introduces a proof calculus for real-analytic differential-algebraic dynamic logic, enabling correct transformations of differential-algebraic equations.
  Applications include index reductions from differential-algebraic equations to ordinary differential equations.
  The calculus ensures compatibility between differential-algebraic equation proof principles and (differential-form) differential dynamic logic for hybrid systems.
  One key contribution is ghost switching which establishes precise
  conditions that decompose multi-modal systems into hybrid systems,
  thereby correctly hybridizing sophisticated differential-algebraic dynamics.
  The calculus is demonstrated in a proof of equivalence for a Euclidean pendulum to index reduced form.
\end{abstract}
\section{Introduction}

Hybrid systems with explicit (vectorial) ordinary differential
equations (ODEs) \(\D{x}=\genDE{x}\) are important models for
cyber-physical systems \cite{henzinger_theory_1996,branicky_general_1996,davoren_logics_2000} and have been axiomatized in
differential dynamic logic \dL
\cite{DBLP:journals/jar/Platzer08,DBLP:journals/jar/Platzer17,Platzer10,Platzer18},
which is the basis of the theorem prover \KeYmaeraX
\cite{DBLP:conf/cade/FultonMQVP15} that was used to prove
correctness properties of aircraft, train and robotics systems \cite{cleaveland_formally_2022,kabra_verified_2022,lin_safe_2022}.
This paper extends (modern differential-form) differential dynamic
logic \cite{DBLP:journals/jar/Platzer17} with
\emph{differential-algebraic programs (DAPs)}.
DAPs allow the relation of $x$ and its time-derivative
\(\D{x}\) to be defined by a formula of first-order real arithmetic.
In particular, DAPs include ordinary equations given by implicit equations such as \(f(x',x) = 0\). Beyond that, DAPs also support composition using logical operators such as conjunctions \(\land\) and disjunctions \(\lor\), enabling modular and flexible specification of system behavior.
This gives rise to easy combination and extension of systems.

Consider a model of an electric circuit consisting of a resistor \(R\), a capacitor \(C\) and a source \(s\). Each component of the circuit can be given by a local constitutive equation:
\[
(R)\; v_R = R\,i_R,\qquad
(C)\; i_C = C\,v_C',\qquad
(s)\; v_s = 0,
\]
where \(v_\bullet\) and \(i_\bullet\) are the voltage and the current of each component.
Because the three elements are connected in series, Kirchhoff's laws contribute two connection constraints \(i_R = i_C\) and \(v_R + v_C = v_s\).
Putting everything together, the electric circuit is described by the single real-arithmetic formula \[
  v_R = R i_R \land i_C = C v_C' \land v_s = 0 \land i_R = i_C \land v_R + v_C = v_s.
\]
Hence, the complete model emerges by simply conjoining each component's constitutive equation with the corresponding connection constraints.
This illustrates the key advantage of differential-algebraic programs: complex systems are built compositionally from independent component.
The same principle applies in other domains---such as hydraulic networks, chemical-process flowsheets, and robotic mechanisms---where differential-algebraic equations arise naturally.
For that reason, differential-algebraic equations are the mathematical foundation for the popular modeling tool \emph{Modelica} \cite{fritzson_modelica_1998}.
However, the DAPs presented in this paper extend beyond differential algebraic equations by allowing inequality and quantified constraints such as \(x' \leq f(x)\) and \(\exists y (g(x,y)=0 \land x'=f(x))\).
This added expressiveness supports the modelling of uncertainty and of dynamics with non-polynomial right-hand sides.

Yet, the expressive power and flexibility of DAPs comes with added complexities in numerical and symbolic analysis.
DAPs may have hidden constraint that are only revealed when differentiating one of their algebraic constraints.
Consider the \emph{Euclidean pendulum}\footnote{Its equations follow from the Lagrangian
\(\mathcal{L}(x,y,\lambda)=g x - \tfrac{1}{2}(v^{2}+w^{2}) + \frac{\lambda}{2}\bigl(x^{2}+y^{2}-1\bigr).\)} \cite[p.~13]{mattsson_index_1993} --- a classical example of a system of differential-algebraic equations.
Its dynamics are given by the formula
\begin{align*}
  x'=v \land v'=\lambda x \land y'=w \land
  w'=\lambda y - g \land x^2 + y^2 = 1,
\end{align*}
where \( x, y \) are the pendulum's position coordinates, and \( v, w \) are the corresponding velocities.  
The scalar \( \lambda \) is a Lagrange multiplier implicitly constrained by the holonomic condition \( x^2 + y^2 = 1 \), and \( g \) denotes the gravitational constant.  
Suppose we wish to analyze the pendulum's behavior with respect to some safety property.
Although \(x\) and \(y\) are confined to the unit circle, their time
derivatives \(x'=v\) and \(y'=w\) are governed by the equations
\(v'=\lambda x\) and \(w'=\lambda y-g\), so the evolution of
\(x\) and \(y\) depends on the unknown multiplier \(\lambda\).
This multiplier is not a free parameter of the system:
solutions exist only if \(\lambda\) satisfies the hidden constraint \[
  \lambda = gy - (v^2 + w^2).
\]
That relation is not explicit in the original model.
It is only uncovered after repeatedly differentiating the algebraic constraint \(x^{2}+y^{2}=1\) and performing algebraic manipulations.
Without this hidden constraint, we cannot directly analyse the behaviour
of \(x\) and \(y\), because \(\lambda\) would remain undetermined.
While such derivations are feasible for small examples, carrying them
out by hand is tedious and error-prone, and doing so outside a formal
logical framework forfeits any end-to-end guarantee of correctness.
Because hidden constraints are a fundamental feature of DAPs, we need systematic and sound principles to uncover them, and a logic expressive enough to state and prove the resulting properties.

To meet this need, this paper introduces the
\emph{differential-form differential-algebraic dynamic logic} (\dAL).
This new logic gives native support to hybrid programs whose continuous
behaviour is specified by DAPs and
provides a sound proof calculus for reasoning about their dynamic
properties.
This calculus fulfills two roles.
First, it guarantees correctness-preserving transformations of DAPs.
Second, it enables rigorous proofs of system properties --- including
safety and reachability --- such proving that \emph{every} state
reachable by a hybrid system with DAPs satisfies a given formula.

A first precursor, \emph{differential-algebraic dynamic logic}
(\DAL) \cite{DBLP:journals/logcom/Platzer10,Platzer10}, already extended the original, non-differential-form variant of
differential dynamic logic.
The differential-form logic \dAL presented in this paper builds on \DAL
and changes it in several essential ways.
The main contributions of this paper are:
\begin{enumerate}
  \item \emph{Logical foundations for model transformations}:
    This calculus establishes a rigorous basis for differential and algebraic transformations of DAPs  --- an indispensable capability for practical work with DAPs that was not available in \DAL.  
    It also lays the groundwork for future integration with larger-scale
    modelling frameworks.
  \item \emph{Axioms for non-normalized differential constraints}:
    The prior DAL proof calculus required a syntactic normalization of
    DAPs.
    The new axioms remove that prerequisite: they apply directly and
    modularly to \emph{non-normalized} formulas, streamlining proofs and
    simplifying implementation.
  \item \emph{Ghost switching axiom}: We introduce a new \emph{ghost switching axiom} that gives   exact criteria for when a DAP can be split into a hybrid program with discrete mode changes.     This enables a sound structural decomposition of multi-mode DAPs
  and gives formal justification to a transformation that has so far
  been used in the literature without proof \cite{DBLP:conf/hybrid/TanMP22,DBLP:journals/tcad/KabraMP22}.
  \item \emph{Differential-form reasoning}: Our proof calculus employs
  differential forms, linking it to existing differential-invariant
  completeness results \cite{DBLP:journals/jacm/PlatzerT20}.  
  The same completeness guarantee now extends to every DAP that can
  be reduced to an ordinary differential equation.
  \item \emph{Local Hilbert-type proof calculus}: Compared to a prior
    DAL sequent calculus, our formulation represents a significant
    step toward a uniform substitution calculus for \dAL.
    Uniform substitution externalizes side-condition checking, simplifying
    proof manipulation and theorem prover implementations
    \cite{DBLP:journals/jar/Platzer17}.
\end{enumerate}

\section{Differential-algebraic Dynamic Logic}

This section briefly introduces the syntax and semantics of
(differential-form) differential-algebraic dynamic logic. The
focus of this discussion lies on the continuous fragment and the
interested reader is referred to the literature
\cite{DBLP:journals/jar/Platzer17,Platzer18} for a treatment of the discrete fragment and the resulting hybrid systems.

\subsection{Syntax}

The set \(\variables\) contains all variables and for each variable \(x
\in \variables\), there is a differential variable \(x' \in \variables\). By
\(\variables' \subseteq \variables\) we denote the set of all differential
variables. A differential variable $x'$ represents the instantaneous
rate of change of $x$, much like a derivative in calculus. However,
rather than being a function of time explicitly it is treated as an
independent variable within the system. This abstraction allows
differentials to be represented algebraically, making it possible to
manipulate them syntactically.
\begin{definition}[Terms and Formulas \protect{\cite{DBLP:journals/jar/Platzer17,DBLP:books/sp/Platzer18}}]
  The language of \dAL \emph{terms} is generated by the following grammar, where \(x\) is a variable, \(c\) is a rational constant, and \(f\) is a function symbol:
  \begin{equation*}
    e, g ::= x \vl c \vl f(e_1, \dots, e_k) \vl e \cdot g \vl  e +
    g \vl (e)'.
  \end{equation*}
  The language of \dAL \emph{formulas} is generated by the following grammar, where \(x\) is a variable, \(P\) a predicate symbol, and \(\alpha\) a differential-algebraic program (\rref{def:DAP}):
  \begin{equation*}
    F, G ::= e \leq g \vl P(e_1, \dots, e_k) \vl F \land G
    \vl \neg F
    \vl \forall xF \vl [\alpha]F.
  \end{equation*}
\end{definition}

\begin{remark}
  Several logical and modal operators such as \(<, \geq, >, \neq, \lor, \rightarrow, \exists\) and \( \langle \cdot \rangle\) are \emph{definable} in terms of more primitive constructs. For instance, disjunction can be defined as \(A \lor B \equiv \neg (\neg A \land \neg B)\), and the diamond modality as \(\langle \alpha \rangle P \equiv \neg [\alpha] \neg P\).
  Throughout this work, we also adopt vectorial notation: \(\tuple{x} = (x_1, \dots, x_n)\) denotes a tuple of variables. For vector-valued function symbols \(f(\tuple{x}) = (f_1(\tuple{x}), \dots, f_n(\tuple{x}))\), the expression \(\tuple{x} = f(\tuple{x})\) abbreviates the conjunction \(\bigwedge_{i=1}^n x_i = f_i(x_1,\dots,x_n)\).
\end{remark}

\begin{remark}
  For a variable \(x\) and formula \(F\) we write \(x \in F\) to mean that \(x\) is a free variable of that formula (See \cite[Def. 7]{DBLP:journals/jar/Platzer17}).
\end{remark}

\begin{definition}[Programs \protect{\cite{DBLP:journals/jar/Platzer17,DBLP:books/sp/Platzer18}}] \label{def:DAP}
  \emph{Differential-algebraic programs} are generated by the following grammar, where \(x\) is a variable, \(\tuple{x}\) is a tuple of variables, \(e\) is a term, and \(F\) is a
  formula of first-order real arithmetic:
  \begin{equation*}
    \alpha, \beta ::= x:= e \vl ?F \vl
    \daesys{\tuple{x}}{F} \vl \alpha;\beta \vl \alpha \cup \beta \vl \alpha^*.
  \end{equation*}
\end{definition}

\begin{remark}
  In this syntax, an assignment is denoted by \(x:=e\), a test by \(?F\), a nondeterministic choice by \(\alpha \cup \beta\), and a nondeterministic loop by \(\alpha^*\). For any \( n \in \mathbb{N} \), the notation \(\alpha^n\) denotes the \(n\)-fold sequential composition of \(\alpha\).
\end{remark}
\begin{remark}
  In order to disambiguate and flexibilize differential-algebraic
equations, the notation chosen here is explicit about which variables
evolve in the DAP and uses logical operators to combine continuous evolutions
\(\daesys{x,y,z}{\D{x}=f(x,z)\land\D{y}=g(y)\land z=y}\).
\end{remark}
\subsection{Semantics}
A \dAL state is a function \(\omega : \variables \rightarrow \R\) that assigns a real value to each variable in the set of all variables \(\variables\).
The set of states is denoted by \(\states\).
Given a subset \(X \subseteq \variables\), we write \(X^\complement = \variables \setminus X \) for its complement.
For a variable \(x\) and a real value \(r \in \R\), we write \(\omega_x^r\) to denote the state obtained from \(\omega\) by replacing its value at \(x\) with \(r\).
Finally, the partial derivative of a differentiable function \(h\) with respect to \(x\) is written as \(\partialDeriv{x}{h}\).
\begin{definition}[Semantics of terms
  \cite{DBLP:journals/jar/Platzer17}]\label{def:terms}
  The \emph{semantics} \(\omega \llbracket e \rrbracket \in \R\) of a \dAL term \(e\) in state \(\omega \in
  \states\) is inductively as follows:
  \begin{enumerate}
    \item \(\omega \sem{ x } = \omega(x)\) for a variable \(x
      \in \mathcal{V}\),
    \item \(\omega \sem{ f(e_1, \dots, e_k) } =
        f(\omega\sem{ e_1 }, \dots, \omega\sem{ e_k
      })\) for a function $f:\R^k \rightarrow \R$,
    \item \(\omega\sem{ e + g } = \omega \sem{
      e } + \omega \sem{ g } \),
    \item \(\omega\sem{ e \cdot g } = \omega
      \sem{ e } \cdot \omega \sem{ g } \),
    \item \(\omega \sem{ (e)'} = \sum_{x \in
        \variables}\omega(x') \partialDeriv{x}{\omega\sem{ e
      }}\).
  \end{enumerate}
\end{definition}
The differential \((e)'\) characterizes how the value of \(e\)
changes locally in response to variations in its free variables
\(x\). Specifically, it expresses this dependence as a function of
the corresponding differential symbols \(x'\), which represent the
rates of change of \(x\). We note that each term \(e\) can only have finitely many free variables. Therefore, the sum in the definition of \((e)'\) is always finite.

\begin{definition}[Semantics of formulas
  \cite{DBLP:journals/jar/Platzer17}]
  The \emph{semantics} of a \dAL formula \(F\) is the set \(\sem{F} \subseteq \states\) of states
  in which \(F\) is true and defined as:
  \begin{enumerate}
    \item \(\sem{e \leq g} = \{\omega \in \states \mid \omega\sem{e}
    \leq \omega\sem{g}\}\),
    \item \(\sem{P(e_1,\dots,e_n)} = \{\omega \in \states \mid (\omega\sem{e_1},\dots,\omega\sem{e_n}) \in \sem{P}\}\),
    \item \(\sem{\neg F} = \states \setminus \sem{F}\),
    \item \(\sem{F \land G} = \sem{F} \cap \sem{G}\),
    \item \(\sem{\forall x F} = \{\omega \in \states \mid \forall r \in \R : \omega_x^r \in \sem{F} \}\),
    \item \(\sem{[\alpha] F} = \{\omega \in \states \mid \forall (\omega, \nu) \in \sem{\alpha} : \nu \in \sem{F}\}\).
  \end{enumerate}
\end{definition}
The \emph{box modality} \([ \alpha ] F\) is used to express \emph{safety properties}, statements that hold after all possible executions of the system.  
If \([ \alpha ] F\) is true in a given state, it guarantees that no matter how the system evolves according to the program \( \alpha \), the condition \(F\) will always hold.
The \emph{diamond modality} \(\langle \alpha \rangle F\), on the other hand, is used to express reachability properties, statements asserting that some execution of the system can lead to a desired state.
If \(\langle \alpha \rangle F\) holds, then there exists at least one execution of the program \( \alpha \) that reaches a state where \(F\) is true.

\begin{definition}[Real-analytic function]
  On an open set \(U \subset \R\) the function \(f: U \rightarrow \R\) said to be \emph{real-analytic}, if for every \(t_0 \in U\) there exist \(\epsilon > 0\) and coefficients $a_k\in\reals$ such that
  \[f(t) = \sum_{k=0}^{\infty}a_k(t-t_0)^k\]
  for all \(t \in (t_0 - \epsilon, t_0 + \epsilon)\).
\end{definition}

\begin{definition}[Flow] \label{def:flow}
  A mapping \(\flowI\) is called a \emph{flow} in a tuple \(\tuple{x}\), if \(\varphi(\cdot)(\tuple{x})\) is real-analytic on \((a,b)\) and satisfies the following two properties:
  \begin{enumerate}
    \item \(\varphi(t)(\bar{x}') = \frac{d \varphi(t)(\bar{x})}{dt}(t)\) for all \(t \in (a,b)\),
    \item \(\varphi(t) = \varphi(0)\) on \((\bar{x} \cup \bar{x}')^\complement\) for all \(t \in (a,b)\).
  \end{enumerate}
  We write \(\varphi \in \analyticFlows{\bar{x}}{(a,b)}\) to denote the set of flows in \(\bar{x}\) defined on the interval \((a,b)\). For a
  formula \(F\) and interval $I$, we also write \(\varphi(I) \subseteq  \sem{F}\) to denote
  \(\varphi(t) \in \sem{F}\) for all \(t \in I\).
\end{definition}

\begin{remark} \label{rem:duration_zero}
  There always exists a flow \(\varphi \in \analyticFlows{\tuple{x}}{(-\epsilon,\epsilon)}\) that satisfies a given state \(\omega \in \states\) at zero. Define \(\varphi : (-\epsilon, \epsilon) \rightarrow \states\), \(\varphi(t)(\tuple{x}) = \omega(\tuple{x}')t + \omega(\tuple{x})\). We can verify that \(\varphi(0)(\tuple{x}) = \omega(\tuple{x})\) and \(\varphi(0)(\tuple{x}') = \frac{d \varphi(t)(\tuple{x})}{dt}(0) = \omega(\tuple{x}')\).
\end{remark}
\begin{definition}[Semantics of differential-algebraic programs] \label{def:pro}
  The semantics of a \emph{differential-algebraic program} \(\alpha\) is the
  transition relation \(\sem{\alpha}
  \subseteq \states \times \states\) and is defined as follows:
  \begin{flalign*}
    &\sem{\alpha \cup \beta} = \sem{\alpha} \cup \sem{\beta}, \\
    &\sem{\alpha ; \beta} = \sem{\alpha} \circ \sem{\beta},\\
    &\sem{x:=e} = \{(\omega, \omega_x^r) \vl r = \omega\sem{e}\}, \\
    &\sem{?F} = \{(\omega, \omega) \vl \omega \in \sem{F}\}, \\
    &\sem{\alpha^*} = \textstyle{\bigcup_{n \in \N} \sem{\alpha^n}}, \\
    &\sem{\daesys{\tuple{x}}{F}} = \{(\omega, \nu) \vl \text{There
      exist } T \geq 0, \epsilon > 0 \text{ and a flow } \varphi \in \analyticFlows{\tuple{x}}{-\epsilon, T + \epsilon}, \\
      &\quad \text{ with } \varphi([0,T]) \subseteq \sem{F} \text{ and } \omega =
    \varphi(0) \text{ and } \varphi(T) = \nu\}.
  \end{flalign*}
\end{definition}
\begin{remark}
  The semantics of differential-algebraic systems require flows to be real-analytic functions on an open neighborhood of \([0,T]\). In contrast, the differential constraint $F$ only needs to hold on the smaller closed set \([0,T]\). This property is crucial for the soundness of the proof calculus in \rref{sec:calculus}.
\end{remark}
\begin{example}\label{ex:boundary_blowup}
  We illustrate the consequences of defining flows on open neighborhoods through an example. Consider the following two differential-algebraic systems:
  \begin{align*}
      \alpha \equiv \daesys{x}{x'=-1 \land x \geq 0} \text{ and } \beta \equiv \daesys{x,y}{x'=-1 \land y^2 = x \land x \geq 0}.
  \end{align*}
  Fix the initial state \(\omega(x) = 1\), \(\omega(x') = -1\), \(\omega(y) = 1\), and \(\omega(y') = 1\).
  For some \(\epsilon > 0\), a flow \(\varphi_{\alpha} \in \analyticFlows{x}{-\epsilon, 1 + \epsilon}\) of system $\alpha$ is given by \(\varphi_\alpha(t)(x) = 1 - t\). 
  Since \(\varphi_\alpha(\cdot)(x)\) is a polynomial, it extends analytically beyond the interval [0,1], even though the differential constraint only holds on [0,1].

  For system $\beta$, a flow \(\varphi_\beta \in \analyticFlows{x,y}{-\epsilon, 1 + \epsilon}\) is given by \(\varphi_\beta(t)(x) = 1-t\) and \(\varphi_\beta(t)(y) = \sqrt{1 - t}\) for any value \(0\leq T < 1\). However, \(\varphi_\beta(\cdot)(y)\) does not extend to \([0,1]\), because its first derivative exhibits finite-time blow-up (see \rref{fig:flows}):
  \[
  \lim_{t \to 1} \frac{d \varphi_\beta(t)(x)}{dt} = \lim_{t \to 1} -\frac{1}{2\sqrt{1-t}} = -\infty.
  \]
\begin{figure}
  \centering
  \tikzset{
    >=stealth,
    line cap=round,
    myaxis/.style={
      axis lines = left,
      enlargelimits=false,
      ticks = none,
      xmin = 0, xmax = 1,
      ymin = 0,
      axis line style = {thick},
      every axis x label/.style={
        at={(current axis.right of origin)}, anchor=west, node font=\small
      },
      every axis y label/.style={
        at={(current axis.above origin)}, anchor=south, node font=\small
      }
    },
    flow/.style={blue!70!black, very thick}
  }
  \begin{tikzpicture}
    \begin{axis}[
      width  = 0.42\textwidth,
      height = 4cm,
      domain = 0:0.99,
      samples=100,
      xlabel = {$t$},
      ylabel = {$y$}
    ]
      \addplot[flow] {sqrt(1-x)};
    \end{axis}
  \end{tikzpicture}
  \hspace{1cm}
  \begin{tikzpicture}
    \begin{axis}[
      width  = 0.42\textwidth,
      height = 4cm,
      domain = 0:0.99,
      samples=100,
      xlabel = {$t$},
      ylabel = {$y'$}
    ]
      \addplot[flow] {-1/(2*sqrt(1-x))};
    \end{axis}
  \end{tikzpicture}

  \caption{The behaviour $\varphi_\beta(\cdot)$ at $y$ and $y'$ as $t \to 1$.}
  \label{fig:flows}
\end{figure}
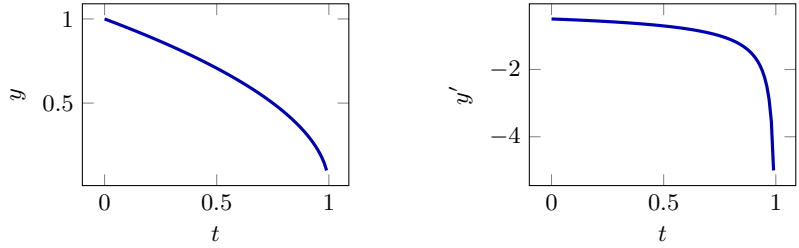

  While for system $\alpha$ and $\beta$ there exists a flow that satisfies the differential constraint each closed subinterval \(I \subset [0,1]\), only system $\alpha$ has a flow that also satisfies it on \([0,1]\).
  More generally, \rref{def:flow} ensures that \(\varphi\) remains well-behaved in a small neighborhood of \([0,T]\) and ensures all derivatives are bounded on \([0,T]\).
\end{example}
Differential-algebraic programs are strictly more expressive than ODE systems in \dL.
Below, we show three key examples:
\begin{enumerate}
    \item \textbf{Modeling Uncertainty}:
    Differential-algebraic programs can represent systems with
    continuous uncertainty by incorporating perturbations.
    For instance, the system
      \[
      \daesys{x, w}{x' = f(x, w) \land w^- \leq
      w \leq w^+}
      \]
    describes a dynamical system where the evolution of \( x \) is
    governed by \( f(x) \) with a disturbance \( w
    \), constrained within the range \( [w^-, w^+] \).

    \item \label{item:implicit}\textbf{Implicit System Characterization}:  
    Differential-algebraic programs can also define systems
    implicitly using existential quantification. Consider:
      \[
      \daesys{x}{\exists y (x' = y \land y^2 = x \land x > 0)}
      \]
    Here, the derivative of \( x \) is characterized by the square
    root of \( x \), as \( x' = y \) must hold for some \( y \)
    satisfying \( y^2 = x \). This implicitly defines the
    nonpolynomial differential equation with nonunique solutions:
      \[
      x' = \pm\sqrt{x}, \quad \text{for } x > 0.
      \]

    \item \textbf{Multimodal Systems}:
    These programs can naturally represent multimodal
    continuous programs. For example, the system
      \[
      \daesys{x}{(x' = f(x) \land x \leq 0) \lor (x' = f(x) \land x \geq 0)}
      \]
    models an evolution in the two domains \(x \leq 0\) and \(x \geq 0\). In \rref{sec:calculus}, we see
    that such programs can be decomposed into a loop with a
    nondeterministic choice between the two modes, provided that the
    handoff points are carefully managed.
    Note that there is no arbitrary switching between these two modes, because flows are real-analytic.
\end{enumerate}
\begin{remark}
  All ODE solutions in the \dL sense are real-analytic (By the Cauchy-Kovalevskaya Theorem \cite[Theorem 1.7.1]{krantz_primer_2002}, all ODE solution with real-analytic right-hand side are real-analytic).
  Thus, the \dAL DAP semantics generalize the \dL ODE semantics, allowing for a broader class of systems.
  In particular, the expressive term language of DAPs in \dAL permits non-Lipschitz differential constraints (See \ref{item:implicit} above).
  Non-Lipschitzian systems exhibit an implicit nondeterminism, such as spontaneous movement from rest \cite{bhatExampleIndeterminacyClassical1997}.
  To avoid unintended implicit behavior, we restrict flows to be real-analytic.
  When nondeterminism is needed, it can be explicitly introduced using discrete program constructs.
  For instance, a flow of the system \(\daesys{x}{x'=f(x) \lor x'=g(x)}\) never switches between the two right-hand sides. Because the flows of both modes are real-analtyic, the Identity Theorem for real-analytic functions \cite[Cor.~1.2.6 p.~14]{krantz_primer_2002} guarantees that any two flows agreeing on an open time interval must agree along their whole trajectories.
  If switching is desired, it must be made explicit with discrete assignments: \[
  \bigl(
     \daesys{x}{x' = f(x)}; x' := g(x)
     \cup
     \daesys{x}{x' = g(x)}; x' := f(x)
  \bigr)^{*}.
\]
  Choosing flows to be real-analytic strikes the balance between flexible modeling and control over system behavior while still allowing for a powerful proof calculus.
\end{remark}
\section{Axiomatic Proof Calculus} \label{sec:calculus}
In this section, we introduce a proof calculus for differential-algebraic programs from \rref{sec:calculus}.
\subsection{Progress, Exit and Entry formulas}
For the soundness of our proof calculus, it is crucial to understand how differential-algebraic equations evolve locally and under what conditions they transition between different modes.

For discrete-time programs, the notion of \emph{next step} is unambiguous: given a current state, there is a distinct next state at the subsequent time step.  
For instance, in the loop \((x := -x)^*\) a state with \( x > 0 \) is followed by a state with \( x < 0 \); each state has a distinct successor state.
In contrast, continuous-time programs differ fundamentally, because the states evolves of real intervals of time.
There is no smallest positive time instance, so there is no notion of next state --- one can speak only of states that are arbitrarily close in the future.

\begin{example} \label{ex:progress}
  Let us consider the differential-algebraic program
  \[
    \daesys{x}{x' = 1 \land -1 < x \land x \leq 1},
  \]
  started in the initial state \(\omega \in \states\). The unique flow of this system is given by \[\varphi(t)(x) = \omega(x) + t\] as long as \(-1 < \varphi(t)(x) \leq 1\) is satisfied. We consider three cases for the initial condition:
  \begin{enumerate}
    \item[\emph{(i)}]  \emph{Strictly inside the domain:}  
          If $-1 < \omega(x)<1$, a positive time increment $\varepsilon$ always
          exists with $\omega(x)+\varepsilon<1$.  
          Hence the system can advance for every such
          $\varepsilon$, and the flow may continue until $x$ reaches~$1$.
    
    \item[\emph{(ii)}] \emph{On the boundary:}  
          If $\omega(x)=1$, the inequality $\varphi(t)(x) \leq 1$ is already tight.
          Any positive advance would violate it, so the only admissible
          evolution is the \emph{blocked} flow
          $\varphi(t)(x)=\omega(x) + t$ for $t = 0$.
    
    \item[\emph{(iii)}] \emph{Outside the domain:}  
          If $\omega(x)>1$, the initial state breaks the constraint
          $\varphi(t)(x) \leq 1$ --- no execution is possible.
  \end{enumerate}
\end{example}
This example illustrates that there exist states in which the system can evolve and others where further progress is impossible. This distinction motivates the following definition:
\begin{definition}[Local progress formula \protect{\cite[p. 24]{DBLP:journals/jacm/PlatzerT20}}] \label{def:progress}
  Let \(\tuple{y}\) and \(\tuple{y}'\) be two tuples of variables that do not occur in the formula \(F\). The \emph{local progress formula} of the system
  \(\daesys{\tuple{x}}{F}\) is defined as
  \begin{align*}
    &\progress{\daesys{\tuple{x}}{F}} \; \equiv \;
    \exists \tuple{y} \exists \tuple{y}' \left(\tuple{x} = \tuple{y} \land \tuple{x}' = \tuple{y}' \land \langle \daesys{\tuple{x}}{F \lor (\tuple{x} = \tuple{y} \land \tuple{x}' = \tuple{y}')} \rangle (\tuple{x} \neq \tuple{y})\right).
  \end{align*}
\end{definition}
\begin{example}[continued]
  Continuing \rref{ex:progress}, the progress formula \(\progress{\daesys{x}{x' = 1 \land x \leq 1}}\) precisely captures the concept we previously discussed: the characterization of states in which the system can locally evolve, i.e., the set of states in which there exists a flow with positive duration.
  In this particular instance, the progress formula is equivalent to  \( x' = 1 \land -1 < x \land x < 1 \).
\end{example}
Equipped with the notion of progress formulas, we now define \emph{exit formulas}, which identify the state from which a flow may no longer continue, and \emph{entry formulas}, which mark the states from which a mode may start.
Together, these formulas make it possible to describe the exact conditions under which a system leaves one mode and enters the next.
\begin{definition}[Exit and Entry Formulas]\label{def:exit_entry}
  Let \(F^- \equiv [\tuple{x}':=-\tuple{x}']F\) denote the reverse-flow formula of a formula \(F\).
  The \emph{local exit formula} and the
  \emph{local entry formula} of the system \(\daesys{\tuple{x}}{F}\) is defined as
  \begin{align*}
    \exit{\daesys{\tuple{x}}{F}} &\equiv (\progress{\daesys{\tuple{x}}{\neg
    F}} \land F) \lor (\progress{\daesys{\tuple{x}}{F^-}}^- \land
    \neg F), \\
    \enter{\daesys{\tuple{x}}{F}} &\equiv (\progress{\daesys{\tuple{x}}{
    F}} \land \neg F) \lor (\progress{\daesys{\tuple{x}}{\neg
    F^-}}^- \land F).
  \end{align*}
\end{definition}
\begin{remark}
  A state $\omega$ satisfies the \emph{exit formula}, if and only if, one of two situations occurs:
  \begin{enumerate}
    \item[(i)] \emph{Forward flow}: Starting from \(\omega\), there exists a flow of positive duration that satisfies the constraint initially, but enter its negation immediately.
    \item[(ii)] \emph{Reverse flow}: Alternatively, beginning in \(\omega\), there exists a time-reverse flow of positive duration that starts in the negation and enters the constraint.
  \end{enumerate}
  Including the reverse-flow disjunct may seem surprising at first, but it is essential when the system's constraint is \emph{open}. Open constraints do not contain their boundary, so no forward flow can reach the boundary while staying in the constraint. By considering a reverse flow starting form the negation, we can capture the exit points on the boundary.
\end{remark}
\begin{example}[continued]
  \begin{figure}[h]
    \centering
    \begin{tikzpicture}[
        x=2cm,y=2cm,>=stealth,line cap=round,
        admiss/.style={very thick,black},
        flow/.style={->,blue!70!black,line width=0.9pt},
        exitarrow/.style={->,black,line width=0.6pt}]
      \begin{scope}
        \draw[->,thick] (-1.15,0) -- (1.15,0) node[below right] {$x$};
        \draw[->,thick] (0,-0.28) -- (0,1.35) node[above] {$x'$};
        \foreach \p/\lab in {-1/-1,1/1}
          \draw (\p,0) ++(0,-0.04) -- ++(0,0.08) node[below=3pt] {\lab};
        \draw[admiss] (-1,1) -- (1,1);
        \draw (-1,1) circle (0.04);
        \filldraw ( 1,1) circle (0.04);
        \foreach \x in {-0.9,-0.6,...,0.8}
          \draw[flow] (\x,1) -- ++(0.18,0);
        \foreach \x in {-0.45,0.25,0.75}{
          \draw[exitarrow] (\x,1)
            .. controls ++(0.18,0) and ++(0,-0.15)
            .. (\x+0.35,1.25);
        }
        \foreach \x in {-0.65,-0.05,0.65}{
          \draw[exitarrow] (\x,1)
            .. controls ++(0.18,0) and ++(0,0.15)
            .. (\x+0.35,0.55);
        }
      \end{scope}
      \begin{scope}[xshift=5cm]
        \draw[->,thick] (-1.15,0) -- (1.15,0) node[below right] {$x$};
        \draw[->,thick] (0,-0.28) -- (0,1.35) node[above] {$x'$};
        \foreach \p/\lab in {-1/-1,1/1}
          \draw (\p,0) ++(0,-0.04) -- ++(0,0.08) node[below=3pt] {\lab};
        \draw[admiss] (-1,1) -- (1,1);
        \filldraw (-1,1) circle (0.04);
        \draw ( 1,1) circle (0.04);
        \foreach \x in {-0.9,-0.6,...,0.8}
          \draw[flow] (\x,1) -- ++(0.18,0);
        \foreach \x in {-0.25,0.35,0.95}{
          \draw[exitarrow] (\x-0.35,1.25)
            .. controls ++(0,-0.15) and ++(0,0)
            .. (\x,1);
        }
        \foreach \x in {-0.65,-0.05,0.65}{
          \draw[exitarrow] (\x-0.35,0.55)
            .. controls ++(0,0.15) and ++(0,0)
            .. (\x,1);
        }
      \end{scope}
    \end{tikzpicture}
    \caption{%
    Exit and entry formula for the differential-algebraic program
    \(\daesys{x}{x' = 1 \land -1 < x \land x \leq 1}\).
    Left: a trajectory that starts on the constraint segment \(x' = 1 \land -1 < x \le 1\) (blue) can exit the mode by leaving the line
    tangentially (black arrows).  The excluded left endpoint is shown as an
    open circle, the included right endpoint as a filled circle.
    Right: for the entry formula the roles of the endpoints
    are reversed, and trajectories must approach the
    segment from outside before continuing along it.}
    \label{fig:entry-exit-dap}
  \end{figure}
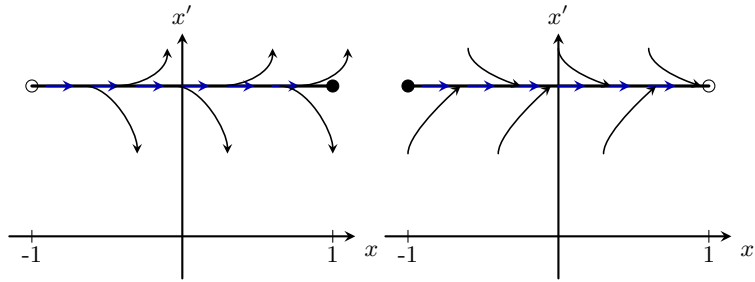
  We now apply the notions of \emph{entry formula} and \emph{exit formula} to the differential-algebraic program
  \[
    \daesys{x}{x' = 1 \land -1 < x \land x \leq 1}
  \]
  from \rref{ex:progress}.
  \emph{Exit formula:}
  A state satisfies the exit formula, if there exist a flow of positive duration that makes progress into the constraint's negation or a reversed flow can progress from the negation into the constraint.
  For the present system this reduces to
  \[
    \exit{\daesys{x}{x' = 1 \land -1 < x \land x \le 1}}
    \;\leftrightarrow\;
    x' = 1 \land -1 < x \land x \le 1 .
  \]
  The left-hand inequality remains strict, because each flow starting in the domain immediately evolves away from the left boundary, i.e., it is no exit point (See \rref{fig:entry-exit-dap} left).
  \emph{Entry formula:}
  In contrast to the exit formula, a flow may enter into the constraint through its left open boundary, even though it is not part of the constraint (See \rref{fig:entry-exit-dap} right). Consequently, both disjuncts from \rref{def:exit_entry} contribute to the entry formula. It reduces to
  \[\enter{\daesys{x}{x' = 1 \land -1 < x \land x \le 1}}
    \;\leftrightarrow\;
    x' = 1 \land -1 \leq x \land x < 1 .
  \]
\end{example}
A central lemma, ensuring the correctness of the definition of the exit and entry formulas, is the following generalization of the intermediate value theorem:
\begin{lemma}
  Let \(T \geq 0, \epsilon > 0\), \(\varphi \in \analyticFlowsInterval\) be a flow and \(F\)
  be a semi-algebraic set. Then, if we have \(\varphi(0) \in \sem{\neg F}\) and \(\varphi(T) \in \sem{F}\), there exists a \(\xi \in [0,T]\) such that \(\varphi(\xi) \in \sem{\enter{\daesys{\tuple{x}}{F}}}\) and \(\varphi(\xi) \in \sem{\exit{\daesys{\tuple{x}}{\neg F}}}\).
\end{lemma}
\begin{proof}
  See \rref{sec:proofs}.
\end{proof}
This lemma formalizes our intuition of entry and exit formulas: when a system's trajectory transitions from a state not satisfying \( F \) to one satisfying \( F \), then there must be an intermediate point simultaneously satisfying the exit formula \(\exit{\daesys{\tuple{x}}{F}}\) and the entry formula \(\enter{\daesys{\tuple{x}}{\neg F}}\).
Now, we finally have the necessary tools to formally define when two modes of a DAP are consistent:
\begin{definition}[Mode consistency]
  Let \(F\) and \(G\) be two formulas of first-order real arithmetic. The \emph{mode consistency formula} of formulas \(F\) and \(G\) is defined as
  \begin{align*}
    &\consis{F}{G} \equiv \left(\exit{\daesys{\bar{x}}{F}} \land
    \enter{\daesys{\bar{x}}{G}}\right) \rightarrow
    (F \leftrightarrow G).
  \end{align*}
\end{definition}
Intuitively, the formula \(\consis{F}{G}\) states that whenever a transition occurs from \( F \) to \( G \), then both constraints must either hold simultaneously or fail to hold together at the transition point.
If this condition is always satisfied along every trajectory of a system, it can be safely decomposed into a two distinct DAPs without losing any of its original behavior.
The importance of this definition for ensuring correct decompositions is illustrated in the following example:
\begin{example} \label{ex:mode_consistency}
  Consider the unrestricted system \[\alpha \equiv \daesys{x}{x' = 1 }.\]
  \emph{Inconsistent split:}
  One might try to decompose the system into the loop \[\beta \equiv \{\daesys{x}{x' = 1 \land x < 0 } \cup  \daesys{x}{x' = 1 \land x \geq 0}\}^*.\]
  The two modes meet at the single boundary state \(x' = 1 \land x = 0\).
  That state is precisely the conjunction of the exit formula of the first mode and the entry formula of the second mode.
  Because the first mode excludes \(x= 0\), it creates an implicit barrier, preventing a continuous transition between the two modes.
  In the original system \(\alpha\), the value of \(x\) evolves continuously from negative to positive values without restriction.
  Hence, \(\alpha\) has distinct behavior from \(\beta\) and the consistency condition fails.
  \emph{Consistent decomposition:} Including the boundary in both modes resolves the issue: \[
  \gamma \equiv \{\daesys{x}{x' =
  1 \land x \leq 0 } \cup \daesys{x}{x' = 1 \land x \geq 0}\}^*.\]
\end{example}
With the notion of mode consistency formally defined, we now present a sound proof calculus for \dAL, which allows rigorous reasoning about mode transitions and model transformations:
\begin{theorem}[Soundness]
  \label{thm:soundness}
  The following formulas are sound axioms of \dAL:
  \begin{flalign*}
    \text{\gray{DW}} &\; \blue{[\daesys{\tuple{x}}{F}]F} \\[4pt]
    \text{\gray{C}} &\; \blue{[\daesys{\tuple{x}}{F \land G}]P} \leftrightarrow  [\daesys{\tuple{x}}{G \land F}]P\\[4pt]
    \text{\gray{DE}} &\; \blue{[\daesys{\tuple{x}}{e = 0}]\der{e}=0} \quad \gray{(\tuple{x}' \not\in e)}\\[4pt]
    \text{\gray{DX}} &\; [\daesys{\tuple{x}}{F}]P
    \rightarrow \blue{[?F]P} \\[4pt]
    \text{\gray{DI}} &\; \left(\blue{[\daesys{\tuple{x}}{F}]P} \leftarrow
    [?F]P\right) \leftarrow \left(F \rightarrow
    [\daesys{\tuple{x}}{F}](P)'\right) \quad \gray{(\tuple{x}' \not\in P)}\\[4pt]
    \text{\gray{DR}} &\; \blue{[\daesys{\tuple{x}}{F}]P} \rightarrow \forall \tuple{y} \forall \tuple{y}'[\daesys{\tuple{x},\tuple{y}}{F \land G}]P \quad \gray{(\tuple{y},\tuple{y}'\not\in P, F)}\\[4pt]
    \text{\gray{AG}} &\; \left(\blue{[\daesys{\tuple{x}}{F}]P}
      \leftarrow \forall \tuple{y} \forall \tuple{y}'[\daesys{\tuple{x},\tuple{y}}{F \land
    G(y)}]P\right) \\
    &\quad \quad \leftarrow [\daesys{\tuple{x}}{F}]G(h(\tuple{x},\tuple{x}')) \quad \gray{(\tuple{y},\tuple{y}'
    \not\in P, F, G(\cdot),h(\cdot,\cdot))}\\[4pt]
    \text{\gray{BDG}} &\; (\blue{[\daesys{\tuple{x}}{F}]P} \leftrightarrow
    [\daesys{\tuple{x},\tuple{y}}{F \land \tuple{y}' = g(\tuple{x},\tuple{x}',y)}]P) \\
    &\quad \quad \leftarrow [\daesys{\tuple{x},\tuple{y}}{F\land \tuple{y}' = g(\tuple{x},\tuple{x}',\tuple{y})
    }]\;\lVert \tuple{y} \rVert^2 \leq h(\tuple{x},\tuple{x}') \quad \gray{(\tuple{y},\tuple{y}' \not\in P, F, h(\cdot,\cdot))}\\[4pt]
    \text{\gray{\(\text{GS}\)}} &\; \left(\blue{[\daesys{\tuple{x}}{F \lor G}]P}
      \leftarrow (F \lor G
        \rightarrow[\{\daesys{\tuple{x}}{F} \cup
    \daesys{\tuple{x}}{G}\}^*]P)\right) \\
    &\quad \quad \leftarrow [\daesys{\tuple{x}}{F \lor G}](\consis{F}{G}
    \land \consis{G}{F})
  \end{flalign*}
\end{theorem}

\begin{proof}
  See \rref{sec:proofs}.
\end{proof}

The axioms for differential-algebraic systems in \rref{thm:soundness} capture essential properties of their evolution.
The \emph{Differential Weakening (DW)} axiom ensures that the formulas representing differential-algebraic equations remain valid throughout the system's execution.
The \emph{Commutativity (C)} axiom allows syntactic exchange of conjunctive formulas.
The \emph{Differential Effect (DE)} axiom justifies that the differential \(\der{e}=0\) of any (potentially differential-algebraic) equation \(e=0\) also holds when following \(e=0\) since equations are closed under differentials.
Thereby, DE also enforces consistency between a variable $x$ and its derivative $\D{x}$, ensuring that $x$ evolves in accordance with the specified dynamics.
The \emph{Differential Skip (DX)} axiom guarantees the existence of a trivial solution with zero duration whenever a formula \( F \) is initially satisfied, meaning the system may remain in place (see \rref{rem:duration_zero}).
The \emph{Differential Invariant (DI)} axiom provides a key condition for proving properties under the box modality.
For instance, if \( x \geq 0 \) initially and \( x' \geq 0 \) throughout, then \( x \geq 0\) at all future times, allowing reasoning about system invariants.
The \emph{Algebraic Ghost (AG)} axiom generalizes differential cuts \cite{DBLP:journals/jar/Platzer17,DBLP:journals/logcom/Platzer10} from \dL to encode observable system behavior.
This is particularly useful when applying the DE axiom to derivatives.
Since the side condition prevents DE from being directly applied to differential variables, we can introduce a new dummy variable \( y = x' \) to track the behavior of \( x' \), allowing the reasoning to proceed.
The \emph{Bounded Differential Ghost (BDG)} axiom allows expressing new differential relationships between variables as long as the new variable $y$ has bounds expressible in terms of prior variables.
The \emph{Ghost Switching (GS)} axiom allows modularly decomposing multi-mode differential algebraic systems into a loop with a non-deterministic choice.
This decomposition exploits analyticity and is only sound if the mode consistency formula is preserved along the evolution of the system.
The key property, that GS exploits for soundness is that an analytic flow can only switch finitely many times in and out of a semi-algebraic set in bounded time.
This is captured in the following lemma:
\begin{lemma} \label{lem:seq}
  Let \(T \geq 0 \) and \(\epsilon > 0\). Further, let \(\varphi \in \analyticFlowsInterval\) be a flow and \(F\)
  be a semi-algebraic set. Let \((\tau_k)_{k \in \N} \subseteq
  [0,T]\) be a strictly monotone sequence with \(\varphi(\tau_k) \in
  \sem{F}\) for all \(k \in \N\). Then, there exists a \(k_0 \in \N\)
  such that \(\varphi(\tau) \in \sem{F}\) for all \(\tau \geq
  \tau_{k_0}\) if the sequence is increasing and for all \(\tau \leq
  \tau_{k_0}\) if the sequence is decreasing.
\end{lemma}
The intuition behind this lemma is that if there exists a sequence of distinct time points where the flow satisfies \(F\), then eventually, the system must satisfy \(F\) indefinitely.
Otherwise, the flow would have to switch infinitely many times in and out of \(F\).
\begin{remark}
  \rref{lem:seq} does not hold on unbounded intervals as witnessed by the real-analytic function \(f:\R \rightarrow [-1,1]\), \(f(t) = \sin(t)\).
\end{remark}

\begin{remark}
  The formula \([\daesys{\tuple{x}}{F}][\daesys{\tuple{x}}{F}]P \leftarrow [\daesys{\tuple{x}}{F}]P \) is not a valid axiom of \dAL (despite a corresponding axiom of \dL). Consider, for instance:
  \[
    \daesys{x,t}{(x'=t \land t \leq 0 \lor x'=-t \land t \geq 0) \land t'=1}.
  \]
  There is no smooth transition from the first to the second mode, due to a discontinuity in the second derivative of \(x\). However, it is possible to identify two analytic flows that agree in both \(x\) and \(x'\). 
  More generally, \dAL is expressive enough to capture systems with non-smooth right-hand sides.
\end{remark}

Some derived axioms are quite useful in practice as well.
\emph{Algebraic Refinement (AR)} is for purely (real) algebraic transformations.
If $F$ implies $G$, then all ways of following $F$ also satisfy $P$ if all ways of following the more general $G$ do.
For instance, AR allows rewriting of \([\daesys{x,y}{x' = 1 \land y'x'=1}]P\) to \([\daesys{x,y}{x' = 1 \land y'=1}]P\).
The \emph{Differential Cut (DC)} axiom is one of the standard differential axioms of \dL \cite{DBLP:journals/jar/Platzer17} to equivalently augment a differential-algebraic equation $F$ with a property $G$ it obeys.
It can be derived by AG and does not introduce a new variable each time.
The \emph{Conjunctive Differential Effect ($\land$DE)} axiom is useful to augment the system's differential properties. For examples, we can transform the system \([\daesys{x}{y'= 2y \land x = y}]P\) to \([\daesys{x}{y'= x^2 \land x = y \land x' = y'}]P\), deriving the differential properties of \(x\) from its algebraic relationship to \(y\).
\begin{theorem}
  The following are derived axioms of \dAL:
  \begin{flalign*}
    &\text{\gray{AR}}\quad \left([\daesys{\tuple{x}}{G}]P \rightarrow
    \blue{[\daesys{\tuple{x}}{F}]P}\right)
    \leftarrow \forall \tuple{x} \forall \tuple{x}' (F \rightarrow G) \\[4pt]
    &\text{\gray{DC}}\quad \left([\daesys{\tuple{x}}{F \land G}]P \leftrightarrow
    \blue{[\daesys{\tuple{x}}{F}]P}\right) \leftarrow [\daesys{\tuple{x}}{F}]G\\[4pt]
    &\text{\gray{\(\land\)DE}} \quad \blue{[\daesys{\tuple{x}}{F \land e =
    0}]P} \leftrightarrow
    [\daesys{\tuple{x}}{F \land e = 0 \land (e)' = 0}]P \quad \gray{(\tuple{x}'
    \not\in e)}
  \end{flalign*}
\end{theorem}
\begin{example}
  Based on the derived \(\land\)DE axiom, one might be tempted to assume the following equivalence holds:
  \[
  [\daesys{\tuple{x}}{F \lor e = 0 }]P \leftrightarrow [\daesys{\tuple{x}}{F \lor (e = 0 \land (e)'=0) }]P.
  \]
  However, this equivalence in not valid in general.
  To see why, consider the following differential-algebraic systems:
  \begin{align*}
    &\alpha \equiv \daesys{x}{x'=1} \equiv \daesys{x}{(x'=1 \land x \neq 0) \lor (x' = 1 \land x = 0)}, \\
    &\beta \equiv \daesys{x}{(x'=1 \land x \neq 0) \lor (x' = 1 \land x = 0 \land x' = 0)}.
  \end{align*}

  In system $\alpha$, \(x\) evolves according to \(x' = 1\), allowing continuous progression from \(x < 0\) to \(x > 0\) without obstruction.
  In contrast, system $\beta$ imposes an additional constraint at \(x = 0\), requiring \(x' = 0\) at that point.
  This restriction causes the system to get stuck at the transition between modes, as \(x'\) cannot jump from \(x'=1\) to \(x'=0\).
  Consequently, the behavior of the DAPs $\alpha$ and $\beta$ is not equivalent.
\end{example}
\subsection{Completeness}
Prior work shows that the natural number are definable using differential-algebraic systems \cite{DBLP:conf/tableaux/2007}.
This means that Gödel's incompleteness theorem applies.
However, there are some partial completeness results for differential invariants in dL \cite{DBLP:journals/jacm/PlatzerT20}.
As our proof calculus builds on top of \dL's proof calculus, our calculus inherits this completeness result for systems which are reducible to polynomial ODEs.
\section{Applications}
In this section, we demonstrate the usefulness of our proof calculus through two examples,  focusing on the reduction of the \emph{differentiation index} of DAPs.
The differentiation index of a differential algebraic equation is equal to the minimum number of times we need to differentiate the equation in order to uniquely determine the derivative of each state variable as a function of the other states variables \cite[p.~179, Def. 4]{campbell_index_1995}.
In the first example, we apply AR and DE in a straightforward manner, successively differentiating the algebraic constraints to obtain a reduced system.
In the second example, a state-dependent singularity in the differentiated constraint necessitates careful case destination using the GS axiom.
\begin{example}[continued]
  We now continue our discussion of the Euclidean pendulum from the introduction, showing
  how our \dAL calculus can be used to eliminate the unknown multiplier \(\lambda\).
  Let \(\tuple{x} = (x,y,v,w)\) denote the state variables, and define the differential-algebraic constraint
  \[D(\lambda) \equiv x'=v \land v'=\lambda
    x \land y'=w \land w'=\lambda y -
  g\land x^2 + y^2 = 1.\]
  Our goal is to simplify the \dAL-formula
  \[
    [\daesys{\tuple{x},\lambda}{D(\lambda)}] P,
  \]
  which states that property \(P\) holds after every run of the system.
  Since \dAL only defines the box modality, we use an arbitrary output predicate \(P\) to model equivalence indirectly.

  The multiplier \(\lambda\) appears explicitly in the constraint, but its time-derivative is unspecified.
  To uncover its implicit dynamics, we differentiate the holonomic constraint \(x^2 + y^2 = 1\) twice, thereby obtaining an explicit expression for \(\lambda\).
  \begin{prooftree}
    \AX$\fCenter \vdash [\daesys{\tuple{x},\lambda}{D(\lambda)  \land xv + yw =
    0 \land \lambda = gy-(v^2 + w^2)}]P$
    \LeftLabel{\(\R\),AR}
    \UI$\fCenter \vdash [\daesys{\tuple{x},\lambda}{D(\lambda)  \land xv + yw =
    0 \land x'v + xv' + y'w + yw' = 0}]P$
    \LeftLabel{\(\land\)DE}
    \UI$\fCenter \vdash [\daesys{\tuple{x},\lambda}{D(\lambda)  \land xv
    + yw = 0}]P$
    \LeftLabel{\(\R\),AR}
    \UI$\fCenter \vdash [\daesys{\tuple{x},\lambda}{D(\lambda)  \land 2xx' + 2
    yy' =0}]P$
    \LeftLabel{\(\land\)DE}
    \UI$\fCenter \vdash [\daesys{\tuple{x},\lambda}{D(\lambda) }]P$
  \end{prooftree}
  Finally, we use algebraic refinement and differential refinement
  to remove \(\lambda\) from the formula with \(h(x,y,v,w) = gy -
  (v^2 + w^2)\):
  \begin{prooftree}
    \AX$\fCenter \vdash [\daesys{\tuple{x}}{D(gy-(v^2 + w^2)) \land xv
    + yw = 0}]P$
    \LeftLabel{DR}
    \UI$\fCenter \vdash \forall \lambda \forall \lambda' [\daesys{\tuple{x},\lambda}{D(h(x,y,v,w)) \land xv + yw = 0 \land \lambda = h(x,y,v,w)}]P$
    \LeftLabel{\(\forall\)i, AR}
    \UI$\fCenter \vdash [\daesys{\tuple{x},\lambda}{D(\lambda) \land xv + yw = 0 \land \lambda = gy-(v^2 + w^2)}]P$
  \end{prooftree}
  The open goal contains the transformed system, in which
  the multiplier \(\lambda\) has been eliminated.
  The result is a fully specified ODE, enabling the application of completeness results for differential invariants \cite{DBLP:journals/jacm/PlatzerT20}.
\end{example}
Note, that we did not need to use GS in the previous example, because we were able to find a globally defined explicit expression for \(\lambda\). We see that this is not always the case in our next example:
\begin{example}
  Consider what appears to be a simplified version of the system:
  \[
  \daesys{x,y}{x' = y \land x^2 + y^2 = 1}.
  \]
  The variable \(x\) evolves according to the value of \(y\), and both are constrained to lie on a circle of radius one. Our goal is to derive an explicit expression for \(y'\). We start by differentiating the circle constraint using \(\land\)DE and do some equivalence transformations with AR:
  \begin{prooftree}
    \AX$\fCenter\vdash [\daesys{x,y}{x' = y \land x^2 + y^2 = 1 \land (x + y')y = 0}]P$
    \LeftLabel{\(\R\), AR}
    \UI$\fCenter\vdash [\daesys{x,y}{x' = y \land x^2 + y^2 = 1 \land xx' + y'y = 0}]P$
    \LeftLabel{AR,\(\R\),\(\land\)DE}
    \UI$\fCenter\vdash [\daesys{x,y}{x' = y \land x^2 + y^2 = 1}]P$
  \end{prooftree}
  Now, we observe that there are two cases for the constraint: \(x + y' = 0\) or \(y = 0\). We split the system into a disjunction of two modes using AR:
  \begin{prooftree}
    \AX$\fCenter\vdash [\daesys{x,y}{(x' = 0 \land x^2=1 \land y = 0) \lor (x' = y \land x^2 + y^2 = 1  \land y'=-x)}]P$
    \LeftLabel{\(\R\), AR}
    \UI$\fCenter\vdash [\daesys{x,y}{x' = y \land x^2 + y^2 = 1 \land (y' = -x \lor y = 0)}]P$
    \LeftLabel{\(\R\), AR}
    \UI$\fCenter\vdash [\daesys{x,y}{x' = y \land x^2 + y^2 = 1 \land (x + y')y = 0}]P$
  \end{prooftree}
  
  We can further decompose this system using the GS axiom. First, we define:
  \begin{align*}
    F(x,y) &\equiv x' = 0 \land x^2=1 \land y = 0, \\
    G(x,y) &\equiv x' = y \land x^2 + y^2 = 1 \land y'=-x.
  \end{align*}
  Here, $F(x,y)$ represents the special case where \(x\) is at its extrema (\(\pm1\)) and remains constant, while $G(x,y)$ describes the general evolution of the system along the unit circle.  
  Now, applying GS we have:
  \begin{prooftree}
    \AX$\fCenter \vdash [\{\daesys{x,y}{F(x,y) \lor G(x,y)}\}](\consis{F(x,y)}{G(x,y)} \land \consis{G(x,y)}{F(x,y)})$
    \noLine
    \UI$\fCenter \vdash [\{\daesys{x,y}{F(x,y)} \cup \daesys{x,y}{G(x,y)}\}^*]P$
    \LeftLabel{GS}
    \UI$\fCenter \vdash [\daesys{x,y}{(x' = 0 \land x^2=1 \land y = 0) \lor (x' = y \land x^2 + y^2 = 1 \land y'=-x)}]P$
  \end{prooftree}
  In \rref{sec:examples}, we show a proof of the top goal by reasoning about progress exit and entry formulas.
  The second goal involves a simplified version of the original system, formulated as a hybrid program consisting of two ODE systems derived from the original DAP.
  This representation clarifies how singularities in the differentiated constraints give rise to distinct modes of system evolution, effectively capturing the different dynamic behaviors that emerge under these conditions.
  To further illustrate these dynamics, \rref{fig:sin_cos} provides a visualization of the system in phase space, showcasing its trajectory and behavior over time.  
  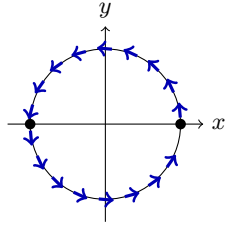
\begin{figure}
    \centering
    \begin{tikzpicture}
      \draw[->] (-1.3,0) -- (1.3,0) node[right] {$x$};
      \draw[->] (0,-1.3) -- (0,1.3) node[above] {$y$};
      
      \draw (0,0) circle (1cm);
      
      \fill[black] (1,0) circle (2pt);
      \node[anchor=west] at (1,0) {};
      
      \fill[black] (-1,0) circle (2pt);
      \node[anchor=east] at (-1,0) {};
      
      \foreach \angle in {5,25,...,335} {
        \draw[->, blue!70!black, very thick] 
             ({cos(\angle)},{sin(\angle)}) 
             -- ++({-sin(\angle)*0.2},{cos(\angle)*0.2});
      }
    \end{tikzpicture}
    \caption{Phase portrait on the unit circle with two stationary modes.}
    \label{fig:sin_cos}
  \end{figure}
\end{example}
\begin{remark}  
  The pattern observed in the previous example can be formulated in a more general way:
  Given a semi-explicit system with of the form  
  \[  
  \daesys{\tuple{x}}{\tuple{x}' = f(\tuple{x}, \tuple{y}) \land g(\tuple{x},\tuple{y}) = 0},
  \]  
  we can transform it using \(\land\)DE into the equivalent system:  
  \[  
  \daesys{\tuple{x}}{\tuple{x}' = f(\tuple{x}, \tuple{y}) \land g(\tuple{x},\tuple{y}) = 0 \land \partialDeriv{x}{g}(\tuple{x},\tuple{y}) f(\tuple{x}) + \partialDeriv{y}{g}(\tuple{x},\tuple{y}) y' = 0}.  
  \]  
  This transformation introduces an implicit constraint on \( \tuple{y}' \), but solving for \( \tuple{y}' \) is only possible if \( \partialDeriv{y}{g}(\tuple{x},\tuple{y}) \) is non-singular. 
  When this condition fails, the system encounters a singularity, highlighting necessity of the GS axiom to distinguish between different cases.
\end{remark}

\section{Related Work}

Differential dynamic logic \dL has established the logical foundations for reasoning about hybrid systems \cite{DBLP:conf/lics/Platzer12a,DBLP:journals/jar/Platzer17}.
However, \dL does not support reasoning about differential-algebraic equations, which are natural for modeling many physical and engineered processes. Building on this gap, an extension of \dL to differential-algebraic equation systems was proposed with a sequent calculus \cite{DBLP:journals/logcom/Platzer10}.
The correct implementation of this approach is hindered by its necessity to globally transform differential-algebraic equations into a specific normal form.
A challenge in reasoning about differential-algebraic systems is understanding the conditions under which systems can be decomposed correctly into different modes for a separate analysis.
Here we exploit entry and exit formulas, which have been used for a complete axiomatization of differential equation invariance \cite{DBLP:journals/jacm/PlatzerT20} and for the characterization of exit sets as a means to decide differential invariants \cite{ghorbal_characterizing_2022}.
Further contributions by Mattsson \cite{mattsson_index_1993}, Pantelides \cite{pantelides_consistent_1988}, and Benveniste \cite{benveniste_structural_2017} have addressed the challenges inherent in simulating differential-algebraic systems.
These works concentrate on the transformation of differential-algebraic equations to ensure consistent initialization.
Although effective for their intended purpose, these approaches are largely confined to the realm of model transformation and simulation.
Other works, such as \cite{cardelli_symbolic_2016,cardelli_erode_2017}, focus on the minimization of polynomial ODE systems using bisimulation techniques.
In contrast, \dAL is strictly more expressive, as it enables the specification and verification of dynamic properties such as safety and liveness.

\section{Conclusion}

This paper introduced a sound Hilbert-style proof calculus for differential-algebraic dynamic logic \dAL, enabling a firm logical foundation for reasoning about differential-algebraic equations.
This allows more flexible specification of continuous-time models, along with a framework for their rigorous transformations to obtain logically decidable invariants.
The Euclidean pendulum example showed that \dAL can specify systems with undetermined dynamics and reduce them to ordinary differential programs.
Future work will be to build on the sound logical reductions that \dAL provides for automatic model transformation.

\section*{Acknowledgements} This work was funded by an Alexander von Humboldt Professorship.
\renewcommand{\doi}[1]{doi: \href{https://doi.org/#1}{\nolinkurl{#1}}}
\bibliographystyle{splncs04}
\bibliography{dae,platzer}

\begin{thebibliography}{10}
\providecommand{\url}[1]{\texttt{#1}}
\providecommand{\urlprefix}{URL }
\providecommand{\doi}[1]{https://doi.org/#1}

\bibitem{benveniste_structural_2017}
Benveniste, A., Caillaud, B., Elmqvist, H., Ghorbal, K., Otter, M., Pouzet, M.: Structural {Analysis} of {Multi}-{Mode} {DAE} {Systems}. In: Proceedings of the 20th {International} {Conference} on {Hybrid} {Systems}: {Computation} and {Control}. pp. 253--263. ACM, Pittsburgh Pennsylvania USA (Apr 2017). \doi{10.1145/3049797.3049806}

\bibitem{bhatExampleIndeterminacyClassical1997}
Bhat, S.P., Bernstein, D.S.: Example of indeterminacy in classical dynamics  \textbf{36}(2),  545--550. \doi{10.1007/BF02435747}

\bibitem{branicky_general_1996}
Branicky, M.S.: General hybrid dynamical systems: {Modeling}, analysis, and control. In: Alur, R., Henzinger, T.A., Sontag, E.D. (eds.) Hybrid {Systems} {III}. pp. 186--200. Springer, Berlin, Heidelberg (1996). \doi{10.1007/BFb0020945}

\bibitem{campbell_index_1995}
Campbell, S.L., Gear, C.W.: The index of general nonlinear {DAEs}. Numerische Mathematik  \textbf{72}(2),  173--196 (Dec 1995). \doi{10.1007/s002110050165}

\bibitem{cardelli_symbolic_2016}
Cardelli, L., Tribastone, M., Tschaikowski, M., Vandin, A.: Symbolic computation of differential equivalences. In: Proceedings of the 43rd {Annual} {ACM} {SIGPLAN}-{SIGACT} {Symposium} on {Principles} of {Programming} {Languages}. pp. 137--150. ACM, St. Petersburg FL USA (Jan 2016). \doi{10.1145/2837614.2837649}

\bibitem{cardelli_erode_2017}
Cardelli, L., Tribastone, M., Tschaikowski, M., Vandin, A.: {ERODE}: {A} {Tool} for the {Evaluation} and {Reduction} of {Ordinary} {Differential} {Equations}. In: Legay, A., Margaria, T. (eds.) Tools and {Algorithms} for the {Construction} and {Analysis} of {Systems}. pp. 310--328. Springer, Berlin, Heidelberg (2017). \doi{10.1007/978-3-662-54580-5_19}

\bibitem{chicone_ordinary_2024}
Chicone, C.: Ordinary {Differential} {Equations} with {Applications}, Texts in {Applied} {Mathematics}, vol.~34. Springer International Publishing, Cham (2024). \doi{10.1007/978-3-031-51652-8}

\bibitem{cleaveland_formally_2022}
Cleaveland, R., Mitsch, S., Platzer, A.: Formally {Verified} {Next}-generation {Airborne} {Collision} {Avoidance} {Games} in {ACAS} {X}. ACM Trans. Embed. Comput. Syst.  \textbf{22}(1),  10:1--10:30 (Oct 2022). \doi{10.1145/3544970}

\bibitem{davoren_logics_2000}
Davoren, J., Nerode, A.: Logics for hybrid systems. Proceedings of the IEEE  \textbf{88}(7),  985--1010 (Jul 2000). \doi{10.1109/5.871305}

\bibitem{fritzson_modelica_1998}
Fritzson, P., Engelson, V.: Modelica — {A} unified object-oriented language for system modeling and simulation. In: Jul, E. (ed.) {ECOOP}’98 — {Object}-{Oriented} {Programming}. pp. 67--90. Springer, Berlin, Heidelberg (1998). \doi{10.1007/BFb0054087}

\bibitem{DBLP:conf/cade/FultonMQVP15}
Fulton, N., Mitsch, S., Quesel, J.D., V{\"o}lp, M., Platzer, A.: {KeYmaera X}: An axiomatic tactical theorem prover for hybrid systems. In: Felty, A., Middeldorp, A. (eds.) CADE. LNCS, vol.~9195, pp. 527--538. Springer, Berlin (2015). \doi{10.1007/978-3-319-21401-6_36}

\bibitem{ghorbal_characterizing_2022}
Ghorbal, K., Sogokon, A.: Characterizing positively invariant sets: {Inductive} and topological methods. Journal of Symbolic Computation  \textbf{113},  1--28 (Nov 2022). \doi{10.1016/j.jsc.2022.01.004}

\bibitem{henzinger_theory_1996}
Henzinger, T.: The theory of hybrid automata. In: Proceedings 11th {Annual} {IEEE} {Symposium} on {Logic} in {Computer} {Science}. pp. 278--292 (Jul 1996). \doi{10.1109/LICS.1996.561342}

\bibitem{kabra_verified_2022}
Kabra, A., Mitsch, S., Platzer, A.: Verified {Train} {Controllers} for the {Federal} {Railroad} {Administration} {Train} {Kinematics} {Model}: {Balancing} {Competing} {Brake} and {Track} {Forces}. IEEE Transactions on Computer-Aided Design of Integrated Circuits and Systems  \textbf{41}(11),  4409--4420 (Nov 2022). \doi{10.1109/TCAD.2022.3197690}

\bibitem{DBLP:journals/tcad/KabraMP22}
Kabra, A., Mitsch, S., Platzer, A.: Verified train controllers for the {Federal Railroad Administration} train kinematics model: Balancing competing brake and track forces. {IEEE} Trans. Comput. Aided Des. Integr. Circuits Syst.  \textbf{41}(11),  4409--4420 (2022). \doi{10.1109/TCAD.2022.3197690}

\bibitem{krantz_primer_2002}
Krantz, S.G., Parks, H.R.: A {Primer} of {Real} {Analytic} {Functions}. Birkhäuser, Boston, MA (2002). \doi{10.1007/978-0-8176-8134-0}

\bibitem{DBLP:conf/lics/2012}
Logic in Computer Science (LICS), 2012 27th Annual IEEE Symposium on. IEEE, Los Alamitos (2012)

\bibitem{lin_safe_2022}
Lin, Q., Mitsch, S., Platzer, A., Dolan, J.M.: Safe and {Resilient} {Practical} {Waypoint}-{Following} for {Autonomous} {Vehicles}. IEEE Control Systems Letters  \textbf{6},  1574--1579 (2022). \doi{10.1109/LCSYS.2021.3125717}

\bibitem{mattsson_index_1993}
Mattsson, S.E., Söderlind, G.: Index {Reduction} in {Differential}-{Algebraic} {Equations} {Using} {Dummy} {Derivatives}. SIAM Journal on Scientific Computing  \textbf{14}(3),  677--692 (May 1993). \doi{10.1137/0914043}

\bibitem{DBLP:conf/tableaux/2007}
Olivetti, N. (ed.): Automated Reasoning with Analytic Tableaux and Related Methods, 16th International Conference, TABLEAUX 2007, Aix en Provence, France, July 3-6, 2007, Proceedings, LNCS, vol.~4548. Springer, Berlin (2007)

\bibitem{pantelides_consistent_1988}
Pantelides, C.C.: The {Consistent} {Initialization} of {Differential}-{Algebraic} {Systems}. SIAM Journal on Scientific and Statistical Computing  \textbf{9}(2),  213--231 (Mar 1988). \doi{10.1137/0909014}

\bibitem{DBLP:journals/jar/Platzer08}
Platzer, A.: Differential dynamic logic for hybrid systems. J. Autom. Reas.  \textbf{41}(2),  143--189 (2008). \doi{10.1007/s10817-008-9103-8}

\bibitem{DBLP:journals/logcom/Platzer10}
Platzer, A.: Differential-algebraic dynamic logic for differential-algebraic programs. J. Log. Comput.  \textbf{20}(1),  309--352 (2010). \doi{10.1093/logcom/exn070}

\bibitem{Platzer10}
Platzer, A.: Logical Analysis of Hybrid Systems: Proving Theorems for Complex Dynamics. Springer, Heidelberg (2010). \doi{10.1007/978-3-642-14509-4}

\bibitem{DBLP:conf/lics/Platzer12b}
Platzer, A.: The complete proof theory of hybrid systems. In: LICS  \cite{DBLP:conf/lics/2012}, pp. 541--550. \doi{10.1109/LICS.2012.64}

\bibitem{DBLP:conf/lics/Platzer12a}
Platzer, A.: Logics of dynamical systems. In: LICS  \cite{DBLP:conf/lics/2012}, pp. 13--24. \doi{10.1109/LICS.2012.13}

\bibitem{DBLP:journals/jar/Platzer17}
Platzer, A.: A complete uniform substitution calculus for differential dynamic logic. J. Autom. Reas.  \textbf{59}(2),  219--265 (2017). \doi{10.1007/s10817-016-9385-1}

\bibitem{Platzer18}
Platzer, A.: Logical Foundations of Cyber-Physical Systems. Springer, Cham (2018). \doi{10.1007/978-3-319-63588-0}

\bibitem{DBLP:books/sp/Platzer18}
Platzer, A.: Logical Foundations of Cyber-Physical Systems. Springer (2018). \doi{10.1007/978-3-319-63588-0}

\bibitem{DBLP:journals/jacm/PlatzerT20}
Platzer, A., Tan, Y.K.: Differential equation invariance axiomatization. J. ACM  \textbf{67}(1),  6:1--6:66 (2020). \doi{10.1145/3380825}

\bibitem{DBLP:conf/hybrid/TanMP22}
Tan, Y.K., Mitsch, S., Platzer, A.: Verifying switched system stability with logic. In: Bartocci, E., Putot, S. (eds.) Hybrid Systems: Computation and Control (part of CPS Week 2022), HSCC'22. ACM (2022). \doi{10.1145/3501710.3519541}

\bibitem{walter_ordinary_1998}
Walter, W.: Ordinary {Differential} {Equations}, Graduate {Texts} in {Mathematics}, vol.~182. Springer, New York, NY (1998). \doi{10.1007/978-1-4612-0601-9}

\end{thebibliography}
\appendix
\section{Additional proof rules and axioms}
In this section, we present additional proof rules omitted from the main discussion. Since the definition of the semantics of the discrete fragment of \dL remains unchanged, the corresponding axioms and proof rules from prior work \cite{DBLP:conf/lics/Platzer12b} continue to apply.
\begin{theorem}
  The following are sound axioms and proof rules of \dAL:
  \begin{mathpar}
    \gray{\left[\cup\right]} \; \blue{[\alpha \cup \beta]P} \leftrightarrow [\alpha]P \land [\beta]P
    \quad
    \gray{\text{K}} \; \blue{[\alpha](R \rightarrow P)} \rightarrow ([\alpha]R
    \rightarrow [\alpha]P)
    \and
    \gray{\forall i} \; \forall x p(x) \rightarrow p(f)
    \and
    \AXC{\(\Gamma_0 \vdash P\)}
    \LeftLabel{\gray{G}}
    \RightLabel{\(\gray{\FV{\Gamma_0}\cap\BV{\alpha} = \varnothing}\)}
    \UIC{\(\Gamma \vdash [\alpha]P\)}
    \DisplayProof
    \and
    [;] \;\blue{[\alpha;\beta]P }\leftrightarrow [\alpha][\beta]P
    \quad
    [:=] \; \blue{[x:=e]P(x)} \leftrightarrow P(e)\\

    \gray{c'} \; \blue{(c)'} = 0 \\
    \gray{x'} \;\blue{(x)'} = x' \\
    \gray{+'} \;\blue{(e + k)'} = (e)' + (k)' \\
    \gray{\cdot'} \;\blue{(e \cdot k)'} = (e)'\cdot k + e \cdot (k)' \\
    \gray{\circ'} \;[y:=g(x)][y':=1](\blue{f(g(x))'} = (f(y))'\cdot (g(x))')
  \end{mathpar}
\end{theorem}
\begin{proof}
  These axioms schemata and proof rules correspond to the axiom of differential dynamic logic \cite{DBLP:conf/lics/Platzer12b}. Therefore, they are sound.
\end{proof}
In the proofs of derived axioms we make use of some additional propositional proof rules:
\begin{mathpar}
  \AXC{\(\Gamma, P \vdash Q, \Delta\)}
  \LeftLabel{{\(\rightarrow\)R}}
  \UIC{\(\Gamma \vdash P \rightarrow Q, \Delta\)}
  \DisplayProof
  \and
  \AXC{\(\Gamma \vdash P, \Delta\)}
  \AXC{\(\Gamma, Q \vdash \Delta\)}
  \LeftLabel{{\(\rightarrow\)L}}
  \BIC{\(\Gamma, P \rightarrow Q \vdash \Delta\)}
  \DisplayProof
  \and
  \AXC{\(\Gamma, P, Q \vdash \Delta\)}
  \LeftLabel{{\(\land\)L}}
  \UIC{\(\Gamma, P \land Q \vdash \Delta\)}
  \DisplayProof
  \and
  \AXC{\(\Gamma \vdash P, \Delta\)}
  \AXC{\(\Gamma \vdash Q, \Delta\)}
  \LeftLabel{{\(\land\)R}}
  \BIC{\(\Gamma \vdash P \land Q, \Delta\)}
  \DisplayProof
  \and
  \AXC{\(\Gamma \vdash P, Q, \Delta\)}
  \LeftLabel{{\(\lor\)R}}
  \UIC{\(\Gamma \vdash P \lor Q, \Delta\)}
  \DisplayProof
  \and
  \AXC{*}
  \LeftLabel{{id}}
  \UIC{\(\Gamma, P \vdash P, \Delta\)}
  \DisplayProof
  \and
  \AXC{\(\Gamma \vdash \Delta\)}
  \LeftLabel{{WL}}
  \UIC{\(\Gamma, P \vdash \Delta\)}
  \DisplayProof
  \and
  \AXC{\(\Gamma \vdash p(y), \Delta\)}
  \LeftLabel{{\(\forall\)R}}
  \UIC{\(\Gamma \vdash \forall x p(x), \Delta\)}
  \DisplayProof

  \AXC{\(\Gamma, p(y) \vdash \Delta\)}
  \LeftLabel{{\(\exists\)L}}
  \UIC{\(\Gamma, \exists x p(x) \vdash \Delta\)}
  \DisplayProof

  \AXC{\(\Gamma \vdash p(e), \Delta\)}
  \LeftLabel{{\(\exists\)R}}
  \UIC{\(\Gamma \vdash \exists x p(x), \Delta\)}
  \DisplayProof

  \AXC{\(\Gamma, P \vdash Q, \Delta\)}
  \AXC{\(\Gamma, Q \vdash P, \Delta\)}
  \LeftLabel{{\(\leftrightarrow\)}}
  \BIC{\(\Gamma \vdash P \leftrightarrow Q, \Delta\)}
  \DisplayProof

  \AXC{\(\Gamma \vdash C, \Delta\)}
  \AXC{\(\Gamma, C \vdash \Delta\)}
  \LeftLabel{{cut}}
  \BIC{\(\Gamma \vdash \Delta\)}
  \DisplayProof
\end{mathpar}

\section{Proof Calculus Soundness} \label{sec:proofs}
In this section, we present the soundness proof for the calculus in \rref{thm:soundness} and some technical lemmas.
\begin{lemma}[Coincidence for terms and formulas] \label{lem:coincidence}
  The set of free variables of terms \(\FV{e}\) and formulas
  \(\FV{\phi}\) are defined in \cite{DBLP:journals/jar/Platzer17}.
  Coincidence property:
  \begin{enumerate}
    \item If states \(\omega, \nu\) agree on the free variables of
      \(e\), then \(\omega\sem{e}=\nu\sem{e}\)
    \item If states \(\omega, \nu\) agree on the free variables of
      \(\phi\), then \(\omega\sem{\phi}\) iff \(\nu\sem{\phi}\).
  \end{enumerate}
\end{lemma}

\begin{lemma}[Properties of Real-Analytic Functions] \label{lem:analytic}
  Let \( f \in C^{\omega}(U) \) and \( g \in C^{\omega}(V) \) be real-analtyic functions, where \( U, V \subseteq \mathbb{R}^n \) are open sets.
  Then the following hold:
  \begin{enumerate}
    \item \( f + g \in C^{\omega}(U \cap V) \),
    \item \( f \cdot g \in C^{\omega}(U \cap V) \),
    \item For any multi-index \( \gamma \), \(\frac{\partial^\gamma
      f}{\partial x^\gamma} \in C^{\omega}(U),\)
    \item If \( g: V \to U \) is a real-analytic function between
      open sets \( V \subseteq \mathbb{R}^m \) and \( U \subseteq
      \mathbb{R}^n \), with \( f: U \to \mathbb{R} \), then the
      composition \( f \circ g \) is in \( C^{\omega}(V) \).
  \end{enumerate}
\end{lemma}

\begin{proof}
  The proofs of 1 and 2 can be found in \cite[Proposition 1.1.4]{krantz_primer_2002}, that of 3 in \cite[Proposition 1.6.3]{krantz_primer_2002}, and that of 4 in \cite[Proposition 1.6.7]{krantz_primer_2002}.
\end{proof}
\begin{theorem}(Identity Theorem \cite[Cor.~1.2.6 p.~14]{krantz_primer_2002})
  \label{thm:identity_theorem}
  Let \(U\) be an open interval and let \(f,g \in C^{\omega}(U)\) be real analytic functions. Then, if
  there is a strictly monotone sequence \((x_k)_{k\in\N}\) in \(U\) with
  \(\lim_{n \rightarrow \infty}x_n \in U\) such that
  \begin{equation*}
    f(x_n) = g(x_n), \text{ for } n = 1,2,\dots
  \end{equation*}
  then
  \begin{equation*}
    f(x) = g(x) \text{ for all } x \in U.
  \end{equation*}
\end{theorem}
\begin{lemma} \label{lem:zeros}
  Let \(U \subseteq \R\) be an open interval and \(f \in C^{\omega}(U)\) be a real-analytic function.
  If \(f\) is non-constant, then \(f\) has at most finitely many zeros in each closed interval \([a,b] \subset U\).
\end{lemma}
\begin{proof}
  Let \(f \in C^{\omega}(U)\) be a non-constant, real-analytic function. Define the zero set \(Z = \{t \in [a,b] \vl f(t) = 0\}\). Assume for a contradiction, that \(Z\) contains infinitely many elements. Fix any element \(t_1 \in Z\).
  Then, either \([a,t_1] \cap Z\) or \([t_1,b] \cap Z\) has an infinite set of elements. Without loss of generality, construct a strictly increasing sequence
  \((t_k)_{k \in \N}\) of elements of \([t_1,b] \cap Z\).
  By the monotone convergence theorem, this sequence has a limit \(\lim_{k
  \rightarrow \infty} t_k = t^*\), which is contained in \(t^* \in
  [a,b]\), because \([a,b]\) is a compact set.
  In summary, we have \(f(t_k) = 0 \) for \(k \in \N\) and \(
  \lim_{k \rightarrow \infty}t_k \in [0,T]\). By \rref{thm:identity_theorem}, this implies \(f(t) = 0\) for all \(t \in U\). This contradicts, the assumption that \(f\) is a non-constant function.
\end{proof}
\begin{remark} \label{rem:inf}
  Define the set \(\mathcal{T} = \{s \vl \varphi(s) \in \sem{F}\} \neq \varnothing\).
  For the infimum \(t = \inf \mathcal{T}\) we have \(t \leq s\) for all \(s \in [0,T]\) with \(\varphi(s) \in \sem{F}\).
  For all \(\tau \in (t,T]\) with \(\varphi(\tau) \in \sem{F}\) there exists an \(s \in [t,\tau)\) such that \(\varphi(s) \in \sem{F}\).
  If \(\varphi(t) \in \sem{\neg F}\), then there exists a \(T-t\geq \delta_0 > 0\) such that \(\varphi((t,t+\delta_0)) \subseteq \sem{F}\) by \rref{cor:zeros}.
\end{remark}
\begin{lemma} \label{lem:local_progress}
  We have \(\omega \in \sem{\progress{\daesys{\tuple{x}}{F}} \lor (\tuple{x}'=0 \land F)}\), if and only if, there
  exist \(\epsilon > 0, T > 0\) and a flow \(\varphi \in \analyticFlowsInterval\) such that \(\varphi((0,T]) \subseteq \sem{F}\) and \(\omega = \varphi(0)\).
\end{lemma}
\begin{proof}[of \rref{lem:local_progress}]
  Fix an arbitrary state with \begin{align*}
    \omega &\in \sem{\progress{\daesys{\tuple{x}}{F}}} \\
    &= \sem{\exists \tuple{y} \exists \tuple{y}' (\tuple{x} = \tuple{y} \land \tuple{x}' = \tuple{y}' \land \dia{\daesys{\tuple{x}}{F \lor (\tuple{x} = \tuple{y} \land \tuple{x}' = \tuple{y}')}}\tuple{x} \neq \tuple{y})}.
  \end{align*}
  Then, there exists a state \(\nu \in \states\) with \(\nu(\tuple{x}) = \nu(\tuple{y}) = \omega(\tuple{x})\) and \(\nu(\tuple{x}') = \nu(\tuple{y}') = \omega(\tuple{x}')\) and \(\nu \in \sem{\dia{\daesys{\tuple{x}}{F \lor (\tuple{x} = \tuple{y} \land \tuple{x}' = \tuple{y}')}}\tuple{x} \neq \tuple{y}}\).
  This implies that there exist \(T \geq 0, \epsilon > 0\) and a flow \(\varphi \in \analyticFlowsInterval\) with \(\varphi(0) = \nu\) such that \(\varphi([0,T]) \subseteq \sem{F \lor (\tuple{x} = \tuple{y} \land \tuple{x}' = \tuple{y}')}\) and \(\varphi(T) \in \sem{\tuple{x}\neq \tuple{y}}\).
  This implies \(\varphi(0) \in \sem{F \lor \tuple{x} = \tuple{y} \land \tuple{x}'=\tuple{z}}\). Assume for a contradiction that \(T = 0\).
  This implies \(\varphi(0)(\tuple{x}) \neq \varphi(0)(\tuple{y})\), which contradicts \(\nu(\tuple{x}) = \nu(\tuple{y})\).
  Therefore, we may assume \(T > 0\) in the following.

  Define \(t = \inf\{s \vl \varphi(s) \in \sem{\tuple{x} \neq \tuple{y}}\}\), which is well-defined by \(\varphi(T) \in \sem{\tuple{x} \neq \tuple{y}}\).
  We consider two cases for the infimum:
  \begin{enumerate}
    \item If \(t > 0\), then \(\varphi((0,t)) \subseteq \sem{\tuple{x}=\tuple{y}}\).
    This implies \(\varphi(\tau)(\tuple{x}) = \varphi(\tau)(\tuple{y})\) for all \(\tau \in (0, t)\).
    By \cite[Corollary 1.2.5]{krantz_primer_2002}, this implies \(\varphi(\tau)(\tuple{x}) = \varphi(\tau)(\tuple{y})\) on \((-\epsilon, T + \epsilon)\) contradicting \(\varphi(T)(\tuple{x}) \neq \varphi(T)(\tuple{y})\).
    \item If \(t = 0\), then by \(\varphi(0) \in \sem{\tuple{x} = \tuple{y}}\) and \rref{rem:inf},  there exists a \(\delta_0 >0 \) such that \(\varphi((0,\delta_0)) \subseteq \sem{\tuple{x}\neq \tuple{y}}\).
    This implies \(\varphi((0,\delta_0)) \subseteq \sem{F}\). In particular, for any \(0 < S < \delta_0\) we have \(\varphi((0,S]) \subseteq \sem{F}\).
  \end{enumerate}
  Now, we define a flow \(\psi \in \analyticFlows{\tuple{x}}{(-\epsilon, S + \epsilon)}\) with \(\psi(\tau)(r) = \varphi(\tau)(r)\) for all \(\tau \in (-\epsilon, S + \epsilon)\) and \(r \in \bar{x} \cup \bar{x}'\) and \(\psi(\tau)(r) = \omega\) for all \(\tau \in (-\epsilon,S+\epsilon)\) and \(r \in (\tuple{x} \cup \tuple{x}')^\complement\).
  By \rref{lem:coincidence}, because \(\tuple{y},\tuple{y}' \not\in \FV{F}\), we have \(\psi((0,S]) \subseteq \sem{F}\).
  Otherwise, we have \(\omega \in \sem{\tuple{x}'=0 \land F}\).
  In that case define \(\varphi(\tau) = \omega\) on \((\tuple{x}')^\complement\) and \(\varphi(\tau) = 0\) on \(\tuple{x}'\) for \(\tau \in [0,T]\).

  Conversely, suppose there are \(T \geq 0, \epsilon > 0\) and \(\varphi \in \analyticFlowsInterval\) with \(\varphi(0) = \omega\) and \(\varphi((0,T]) \subseteq \sem{F}\). Consider three disjoint cases:
  \begin{enumerate}
    \item If \(\omega \in \sem{\tuple{x}' = 0 \land F}\), then the statement immediately follows.
    \item If \(\omega \in \sem{\neg F \land \tuple{x}' = 0}\), we have \(\omega = \varphi(0) \in \sem{\neg F}\). By \(\varphi(0) = \omega \in \sem{\tuple{x}' = 0}\), we also have \(\varphi(0)(\tuple{x}) = \varphi(T)(\tuple{x})\). The previous two properties contradict the assumption that there exists a flow \(\varphi \in \analyticFlowsInterval\) with \(\varphi(T) \in \sem{F}\).
    \item If \(\omega \in \sem{\tuple{x}' \neq 0}\). Then, by the continuity of derivatives there exists a \(\delta > 0\) and a variable \(r' \in \tuple{x}'\) such that \(\varphi(\tau)(r') \neq 0 \) for all \(\tau \in [0,\delta)\). By applying the mean-value theorem we obtain:  \(\varphi(\delta)(r) - \varphi(0)(r) = \tau\frac{d\varphi(t)}{dt}(\xi) \neq 0\) for some \(\xi \in (0,\delta)\). Therefore, we have \(\varphi(\delta) \in \sem{\tuple{x} \neq \tuple{y}}\).
  \end{enumerate}
\end{proof}

\begin{proof}[of \rref{lem:seq}]
  For \(T \geq 0\) and \(\epsilon > 0\) let \(\varphi \in \analyticFlowsInterval\) be a flow. Because \(F\) is a semi-algebraic set,
  form, we can represent it by its disjunctive normal: \[
    F \equiv \bigvee_{i=1}^n\left(\bigwedge_{p \in
      \mathcal{P}_i} p \geq 0 \land \bigwedge_{q \in
    \mathcal{Q}_i} q > 0\right),
  \]
  where \(\mathcal{P}_i\) and
  \(\mathcal{Q}_i\) are finite sets of polynomials for \(i \in \N\). For the flow \(\varphi\) it is
  equivalent, to instead consider the reduced representation
  \[
    F_{\mathcal{I}} \equiv \bigvee_{i=1}^n\left(\bigwedge_{p \in
      \mathcal{P}_i \setminus \mathcal{I}} p \geq 0 \land
    \bigwedge_{q \in \mathcal{Q}_i\setminus \mathcal{I}} q > 0\right),
  \]
  where \(\mathcal{I} = \{p \in \Pi \vl \varphi(0)\sem{p} =
  \varphi(t)\sem{p} \text{ for all } 0 \leq t \leq T\}\) is set of
  all polynomials that remain constant along the flow. Without loss
  of generality assume that \((\tau_k)_{k \in \N}\) is strictly
  monotonically increasing. Suppose that for every \(k_0 \in \N\)
  there exists a \(t \geq \tau_{k_0}\) such that \(\varphi(t) \in
  \sem{\neg F_{\mathcal{I}}}\). We know that \(\varphi(\tau_{k_0+1})
  \in \sem{F_{\mathcal{I}}}\). By the continuity of the
  solution, the intermediate value theorem implies that exists an
  intermediate value with \(t
  \leq \xi_{k_0} \leq \tau_{k_0+1}\) and \(\varphi(\xi_{k_0})\sem{p(x,x')} =
  0\) for some \(p \in (\mathcal{P} \cup \mathcal{Q})\setminus
  \mathcal{I}\), where \(\mathcal{P} =
  \bigcup_{i=1}^n \mathcal{P}_i\) and \(\mathcal{Q} = \bigcup_{i=1}^n
  \mathcal{Q}_i\). By repeating this argument we can construct
  a strictly increasing sequence of intermediate values
  \((\xi_{k_0})_{k_0 \in \N}\) such that for every \(k_0 \in \N\)
  there exists a \(p \in (\mathcal{P} \cup \mathcal{Q})\setminus
  \mathcal{I}\) with \(\varphi(\xi_{k_0})\sem{p(x,x')} = 0\).
  Because the set of polynomials \((\mathcal{P} \cup
  \mathcal{Q})\setminus
  \mathcal{I}\) is finite, this implies that there is at least one
  element \(p \in (\mathcal{P} \cup \mathcal{Q})\setminus
  \mathcal{I}\) with an infinite set of zeros in \([0,T]\).
  The value of that polynomial along the flow is given by \(\varphi(\cdot)\sem{p(x,x')}:(-\epsilon,T+\epsilon) \rightarrow \R\) which is real-analytic on \((-\epsilon,T+\epsilon)\) by \rref{lem:analytic}. Because \(p \not\in \mathcal{I}\) is non-constant along the flow, the function \(\varphi(\cdot)\sem{p(x,x')}\) is non-constant. This contradicts \rref{cor:zeros}, because \(\varphi(\cdot)\sem{p(x,x')}\) has an infinite set of zeros in \([0,T]\).
  \qed
\end{proof}

\begin{corollary} \label{cor:zeros}
  Let \(T \geq 0, \epsilon > 0\), \(\varphi \in \analyticFlowsInterval\) be a flow and \(F\)
  be a semi-algebraic set. Then, if for all \(\delta \in (t,T]\)
  there exists a \(\tau \in (t,t+\delta)\) such that \(\varphi(\tau)
  \in \sem{F}\),
  then there exists a \(\delta_0  > 0\) such that for all \(\tau \in (t,t + \delta_0) \) we have \(\varphi(\tau) \in \sem{F}\).
\end{corollary}
\begin{proof}
  Let \(T \geq 0, \epsilon > 0\) and \(\varphi \in \analyticFlowsInterval\) be a flow. Construct a
  strictly deceasing sequence \((\tau_{k})_{k \in \N}\) such that
  \(\varphi(\tau_k) \in \sem{F}\) by \rref{lem:seq} there exists
  a \(k_0 \in \N\) such that \(\varphi(\tau) \in \sem{F}\) for all
  \(\tau \leq \tau_{k_0}\). \qed
\end{proof}
\begin{lemma} \label{lem:progress}
  Let \(T \geq 0, \epsilon > 0\), \(\varphi \in \analyticFlowsInterval\) be a flow and \(F\)
  be a semi-algebraic set. Then, if we have \(\varphi(0) \in \sem{\neg F}\) and \(\varphi(T) \in \sem{F}\), the following holds: The infimum \(t = \inf \{ \tau \vl \varphi(\tau) \in \sem{F}\}\) satisfies \(\varphi(t) \in \sem{\enter{\daesys{\tuple{x}}{F}}}\) and \(\varphi(t) \in \sem{\exit{\daesys{\tuple{x}}{\neg F}}}\).
\end{lemma}
\begin{proof}
  Let \(T \geq 0, \epsilon > 0\) and \(\varphi \in \analyticFlowsInterval\) be a flow with \(\varphi(0) \in\sem{\neg F}\) and \(\varphi(T) \in \sem{F}\). We have \(\varphi(\tau) \in \sem{\neg F}\) for \(\tau \in [0,t)\), by the definition of the infimum. First, we show \(\varphi(t) \in \sem{\enter{\daesys{\tuple{x}}{F}}}\), by splitting into two disjoint cases:
  \begin{enumerate}
    \item Suppose \(\varphi(t) \in \sem{F}\). Then show \(\varphi(t) \in \sem{\progress{\daesys{\tuple{x}}{\neg F^-}^-}}\). Construct \(\psi(\tau)(\tuple{x}) = \varphi(t - \tau)(\tuple{x})\) and \(\psi(\tau)(\tuple{x}') = - \varphi(t - \tau)(\tuple{x}')\) for \(\tau \in [0,t]\). Then we have \(\psi((0,\tau]) \subseteq \sem{\neg F^-}\). This implies \(\psi(0) \in \sem{\progress{\daesys{\tuple{x}}{\neg F^-}}}\), which implies \(
    \varphi(t) = \psi(0)_{\tuple{x}'}^{-\psi(0)(\tuple{x}')} \in \sem{\progress{\daesys{\tuple{x}}{\neg F^-}}^-}\).
    \item Otherwise, suppose \(\varphi(t) \in \sem{\neg F}\). We show \(\varphi(t) \in \sem{\progress{\daesys{\tuple{x}}{F}}}\): If there exists a \(\delta \in (0,T]\) such that for all \(\tau \in (t,t+\delta)\) we have \(\varphi(\tau) \in \sem{F}\), the flow satisfies \(\varphi((t,t+\delta_0]) \subseteq \sem{F}\) for some \(\delta_0 < \delta\). By \rref{lem:local_progress}, this implies the statement. Now, suppose for a contradiction that for every \(\delta \in (0,T]\), there exists a \(\tau \in (t,t+\delta)\) such that \(\varphi(\tau) \in \sem{\neg F}\). By \rref{cor:zeros}, this implies that there exists a \(\delta_0 \in (t,T]\) such that for all \(\tau \in (t,t+\delta_0)\) we have \(\varphi(\tau) \in \sem{\neg F}\). Together with \(\varphi(t) \in \sem{\neg F}\) this implies for all \(\tau \in [t,t+\delta)\) we have \(\varphi(\tau) \in \sem{\neg F}\), contradicting that \(t\) is the infimum of \(\{ \tau \vl \varphi(\tau) \in \sem{F}\}\).
  \end{enumerate}
  The two cases combined imply \(\varphi(t) \in \sem{\enter{\daesys{\tuple{x}}{F}}}\). By \rref{def:progress}, we have 
  \begin{align*}
    \sem{\enter{\daesys{\tuple{x}}{F}}} &= \sem{(\progress{\daesys{\tuple{x}}{
    F}} \land \neg F) \lor (\progress{\daesys{\tuple{x}}{\neg
    F^-}}^- \land F)} \\
    &= \sem{\exit{\daesys{\tuple{x}}{\neg F}}}.
  \end{align*}
  This implies the second statement.
\end{proof}
\begin{lemma} \label{lem:mode_consistency}
  Let \(T \geq 0, \epsilon > 0 \) and \(\varphi \in \analyticFlowsInterval\) be a flow with \(\varphi(0) \in
  \sem{F}\) and \(\varphi([0,T]) \subseteq \sem{F \lor G}\). If we have \(\varphi([0,T]) \subseteq
  \sem{\consis{F}{G}}\) and there exists a \(\xi \in [0,T]\) such that
  \(\varphi(\xi) \in \sem{G}\), then there a \(\tau \in [0,T]\) such that
  \(\varphi([0,\tau]) \subseteq \sem{F}\) and \(\varphi(\tau)
  \in \sem{F \land G}\).
\end{lemma}

\begin{proof}
  Fix \(T \geq 0, \epsilon > 0 \) and \(\varphi \in \analyticFlowsInterval\) with \(\varphi(0) \in
  \sem{F}\) and \(\varphi([0,T]) \subseteq \sem{F \lor G}\) and \(\varphi(\xi) \in \sem{G}\) for some \(\xi \in [0,T]\). If \(\varphi(0)
  \in \sem{G}\), the statement immediately follows. Otherwise,
  consider \(\varphi(0) \in \sem{\neg G}\).
  Define the set \(\mathcal{T} = \{t \vl \varphi(t) \in
  \sem{G}\}\). Choose \(t\) to be the infimum of
  \(\mathcal{T}\), which is well-defined because of \(\varphi(\xi) \in
  \sem{G}\). By the definition of the infimum, we have \(\varphi(\tau) \in
  \sem{\neg G}\) for all \(\tau \in [0,t)\). This implies
  \(\varphi(\tau) \in \sem{F}\) for all \(\tau \in [0,t)\) and
  \(\varphi (t) \in \enter{\daesys{x}{G}}\) by \rref{lem:progress}.
  We consider three separate cases:
  \begin{enumerate}
    \item \(\varphi(t) \in \sem{\exit{\daesys{\tuple{x}}{F}}}\): We show
      \(\varphi(t) \in \sem{F \land G}\): We have
      \(\varphi(t) \in \sem{\exit{\daesys{\tuple{x}}{F}} \land \enter{\daesys{\tuple{x}}{G}}}\).
      This implies \(\varphi(t) \in \sem{F \leftrightarrow
      G}\). By assumption, we have \(\varphi(t) \in
      \sem{F \lor G}\). Without loss of generality, we
      assume \(\varphi(t) \in \sem{F}\). This implies
      \(\varphi(t) \in \sem{F \land G}\).
    \item \(\varphi(t) \in \sem{\neg \exit{\daesys{\tuple{x}}{F}} \land
      F}\): This condition implies \(\varphi (t) \in \sem{\neg
      \progress{\daesys{\tuple{x}}{\neg F}}}\). If \(t = T\), then by assumption this implies \(\varphi(T) \in \sem{G}\) and also \(\varphi(T) \in \sem{F}\).
      Otherwise, consider \(t < T\). By \rref{lem:local_progress}, for all \(S
      > 0, \delta > 0\) and all flows \(\psi \in \analyticFlows{\tuple{x}}{(-\delta, S + \delta)}\) with \(\psi(0) = \varphi(t)\) there exists a \(\tau \in (0,S]\) such that \(\psi(\tau) \in \sem{F}\).
      In particular, this implies if that for all \(\delta \in (t,T]\) there exists a \(\tau \in (t, t + \delta)\) such that \(\varphi(\tau) \in \sem{F}\). By \rref{cor:zeros}, this implies
      that there exists a \(\delta_0 > 0\) such that \(\varphi((t, t +\delta_0)) \subseteq
      \sem{F}\). Similarly,
      there exists a \(t \leq \tau < t +
      \delta_0\) such that \(\varphi(\tau) \in \sem{G}\):
      assume for all \(t \leq \tau < t + \delta_0\) we have
      \(\varphi(\tau) \in \sem{\neg G}\). Then, \(t+\delta_0\)
      is a lower bound \(\mathcal{T}\), but \(t + \delta_0 > t\),
      which contradicts the assumption that \(t\) is the infimum of
      \(\mathcal{T}\). This implies \(\varphi(\tau) \in \sem{F}
      \cap \sem{G} = \sem{F \land G}\). 
    \item \(\varphi(t) \in \sem{\neg \exit{\daesys{\tuple{x}}{F}} \land
      \neg F}\): This condition implies \(\varphi(t) \in
      \sem{\neg\progress{\daesys{\tuple{x}}{F^-}}^-}\). We can easily see
      that this contradicts the existence of \(\varphi\) by
      constructing \(\psi \in \analyticFlowsInterval\) with
      \(\psi (\tau)(\tuple{x}) = \varphi(t-\tau)(\tuple{x})\) and
      \(\psi (\tau)(\tuple{x}') = -\varphi(t-\tau)(\tuple{x}')\)
      for \(\tau \in [0,t]\). For this flow we have \(\psi((0,t]) \subseteq \sem{F^-}\), which contradicts \(\varphi(t) = \psi(0)_{\tuple{x}'}^{-\psi(0)(\tuple{x}')}\in \sem{\neg\progress{\daesys{\tuple{x}}{F^-}}^-}\).
  \end{enumerate}
  In the first two cases, there exists a \(\tau \in [0,T]\) such that \(\varphi(\tau) \in \sem{F \land G}\) and the third case leads to a contradiction.
  \qed
\end{proof}

\begin{lemma}[Differential Lemma] \label{lem:diff}
  Let \(T \geq 0, \epsilon > 0\) and \(\varphi \in \analyticFlowsInterval\) be a flow and \(e\) be a term with \(\FV{e} \subseteq \tuple{x}\).
  Then we have
  \[
    \varphi(t)\sem{(e)'} = \frac{d \varphi(t)\sem{e}}{dt}
  \]
  for all \(t \in [0,T]\).
\end{lemma}
\begin{proof}
  Let \(T \geq 0, \epsilon > 0\) and \(\varphi \in \analyticFlowsInterval\) be a flow and \(e\) a term with \(x'
  \not\in e\). Then, for \(t \in [0,T]\) we have
    \[
    \varphi(t)\sem{(e)'} 
    = \sum_{r \in e} \varphi(t)(r') \frac{\partial \varphi(t)\sem{e}}{\partial r} 
    = \sum_{r \in e} \frac{d \varphi(t)(r)}{d t} \frac{\partial
    \varphi(t)\sem{e}}{\partial r} = \frac{d \varphi(t)\sem{e}}{d t},
    \]
    by differential consistency of \(\tuple{x}\) and \(\tuple{x}'\) and the chain rule.
    \qed
\end{proof}

\begin{proof}
  While some of the axioms in \rref{thm:soundness} resemble those
  in \dL, special care must be taken to ensure their soundness,
  particularly due to the non-uniqueness of solutions and the
  requirement to show that solutions are analytic.

  Key considerations are:
  \begin{enumerate}
    \item DW combines the behavior of \dL's DW and DE
      \cite{DBLP:journals/jar/Platzer17}, if the derivative of all state
      variables is explicit.
    \item DE encodes the differential consistency between \(x\) and
      \(x'\) and was not included in the original calculus of \dL
      \cite{DBLP:journals/jar/Platzer17} in this form.
    \item BDG Differential ghosts need a new argument, because we must
      show that the flow of the augmented system analytic.
    \item AG generalizes differential cuts of \dL \cite{DBLP:journals/jar/Platzer17}.
    \item GS is an entirely new axiom.
  \end{enumerate}
  \begin{itemize}
    \item[DW] Let \(\omega \in \states\) be an arbitrary state. Let \(T \geq 0 \), \(\epsilon > 0\) and 
      \(\varphi \in
      \analyticFlowsInterval\) be a fixed flow with \(\varphi([0,T]) \subseteq \sem{F}\) and
      \(\varphi(0) = \omega\). Then, by \rref{def:flow}, we
      have \(\varphi(t) \in \sem{F}\) for all \(t \in [0,T]\). In
      particular, this implies \(\varphi(T) \in \sem{F}\).
    \item[C] Let \(\omega \in \states\) be a fixed state.
    Assume \(\omega \in \sem{[\daesys{\tuple{x}}{F \land G}]P}\). 
    Then, for all flows \(\varphi \in \analyticFlows{\tuple{x}}{-\epsilon, T+\epsilon}\) with \(\varphi([0,T]) \in \sem{F \land G}\) and \(\varphi(0) = \omega\) we have \(\varphi(T) \in \sem{P}\). Let \(\tilde{\varphi} \in \analyticFlowsInterval\) be a second flow with \(\varphi(0) = \omega\) and \(\varphi([0,T]) \in \sem{G \land F}\).
    By \(\sem{G \land F} = \sem{F \land G}\) we have \(\varphi([0,T]) \in \sem{F \land G}\).
    This implies \(\varphi(T) \in \sem{P}\). This implies \(\omega \in \sem{[\daesys{\tuple{x}}{G \land F}]P}\).
    The converse follows by symmetry.
    \item[DE] Let \(\omega \in \states\) be a fixed state. We need to
      show that \(\omega \in \sem{[\daesys{\tuple{x}}{e = 0}](e)' = 0}\). Let \(T \geq 0 \), \(\epsilon > 0\) and 
      \(\varphi \in
      \analyticFlowsInterval\) be an arbitrary flow with \(\varphi([0,T]) \in \sem{e = 0}\).  By assumption, we have \(\varphi(t)\sem{e} =
      0\) for all \(t \in [0,T]\).
      This implies
      \[\varphi(t)\sem{(e)'} = \frac{d \varphi(t)\sem{e}}{d t} = 0,\]
      by \rref{lem:diff}.
    \item[DX] Let \(\omega \in \sem{[\daesys{\tuple{x}}{F}]P}\).
      Then, for all \(T
      \geq 0\), \(\epsilon > 0\) and flows \(\varphi \in \analyticFlowsInterval\) with \(\varphi([0,T])
      \subseteq \sem{F}\) and \(\varphi(0) = \omega\), we have \(\varphi(T) \in \sem{P}\). To show \(\omega \in
      [?F]P\), we need to show if \(\omega \in \sem{F}\), then
      \(\omega \in \sem{P}\).
      Suppose \(\omega \not\in \sem{F}\), then the formula is trivially true. Otherwise, suppose \(\omega \in \sem{F}\) and choose \(T = 0\).
      Then, by \rref{rem:duration_zero},
      there exists one flow \(\varphi \in \analyticFlows{\tuple{x}}{-\epsilon, \epsilon}\) with \(\varphi(0)
      = \omega\).
      By assumption, this implies \(\omega = \varphi(0)
      \in \sem{P}\).
    \item[BDG] First, we show the implication \(\rightarrow\). For this direction we do not need the assumption of bounded growth: Let
    \(\omega \in \sem{[\daesys{\tuple{x}}{F}]P}\). Assume that \(\omega
    \in \sem{F}\). Otherwise, there is no flow that satisfies the
    differential constraint \(F\). Also, fix arbitrary \(T
    \geq 0, \epsilon > 0\) and a flow \(\varphi \in \analyticFlows{\tuple{x},\tuple{y}}{-\epsilon, T + \epsilon}\) with \(\varphi([0,T])
    \in \sem{F \land \tuple{y}'=g(\tuple{x},\tuple{y})}\) and \(\varphi(0) = \omega\). We need
    to show \(\varphi(T) \in \sem{P}\). In particular, we have
    \(\varphi([0,T]) \subseteq \sem{F}\). By assumption, this implies
    \(\varphi(T) \in \sem{P}\).

    Conversely, we show the implication \(\leftarrow\): Fix an arbitrary state \(\omega \in \sem{[\daesys{x}{F \land y'=g(x,y)}]||y||^2 \leq h(x,x')}\) and \(\omega \in \sem{[\daesys{\tuple{x}}{F \land \tuple{y}'=g(\tuple{x},\tuple{y})}]P}\). If \(\omega \not \in \sem{F}\), there are no \(T \geq 0, \epsilon > 0\) and flow \(\varphi \in \analyticFlowsInterval\) such that \(\varphi([0,T]) \subseteq \sem{F}\) and \(\varphi(0) = \omega\) and there is nothing to show.

    Assuming \(\omega \in \sem{F}\), take an arbitrary \(\varphi \in \analyticFlowsInterval\) with \(\varphi([0,T]) \subseteq \sem{F}\) and \(\varphi(0) = \omega\).
    Now, the idea is to construct a solution of the augmented system out of \(\varphi\) that has an equal interval of existence, and to exploit that \(P\) always holds after running that system.
    Define \(y(\cdot) : (-\delta,T_y) \rightarrow \R^m\) to be the right-maximal solution \cite[\S10.XI]{walter_ordinary_1998} of
    \[
      \frac{d y(t)}{dt}=g(\varphi(t)(\tuple{x}), y(t)), \quad y(0) = \omega(y).
    \]
    Observe that \(g(\varphi(\cdot)(\tuple{x}), \cdot)\) is a real-analytic function on \((-\epsilon,T+\epsilon)\) by \rref{lem:analytic}.
    Then, by the Cauchy-Kowalewsky Theorem \cite[Theorem 1.7.1]{krantz_primer_2002}, the solution is a real-analytic function on \((-\delta,T_y) \subseteq (-\epsilon, T + \epsilon)\).
    Now, we construct the mapping \(\varphi_y : (-\delta,T_y) \rightarrow \states\) with \(\varphi_y(z)(t) = \varphi(z)(t)\) on \(z \in (\tuple{y}\cup \tuple{y}')^\complement\), \(\varphi_y(t)(\tuple{y}) = y(t)\) and for all \(t \in (-\delta,T_y)\) and  \(\varphi_y(t)(\tuple{y}') = \frac{d y(t)}{dt}\) for \(t \in (-\delta,T_y)\).

    Now, we show \(T_y = T\). By assumption, we have 
    \[
      ||\varphi_y(\tau)(\tuple{y})||^2 \leq h(\varphi(\tau)(\tuple{x}), \varphi(\tau)(\tuple{x}')),
    \]
    for \(\tau \in [0,T]\).
    By the continuity of \(\varphi(\cdot)(\tuple{x})\) and \(\varphi(\cdot)(\tuple{x}')\), this implies \(||\varphi_y(\cdot)(\tuple{y})||^2\) attains its maximum on \([0,T]\).
    Therefore, the solution exhibits no blow up in finite time and by \cite[Theorem 1.4]{chicone_ordinary_2024}, this implies the domain of the solution is equal to the domain of definition of \(g(\varphi_y(\cdot)(\tuple{x}),\cdot)\), meaning \(T = T_y\).
    By the construction, the flows \(\varphi_y\) and \(\varphi\) only differ on \(\tuple{y} \cup \tuple{y}'\). By \rref{lem:coincidence}, this implies \(\varphi_y([0,T]) \subseteq \sem{F}\), because \(\tuple{y},\tuple{y}' \not \in F\).
    Therefore, we have \(\varphi_y([0,T]) \subseteq \sem{F \land \tuple{y}'=g(\tuple{x},\tuple{y})}\).
    Finally, by the assumption \(\omega \in \sem{[\daesys{\tuple{x}}{F \land \tuple{y}'=g(\tuple{x},\tuple{y})}]P}\), we have \(\varphi_y(T) \in \sem{P}\). 
    Again, using \rref{lem:coincidence}, implies \(\varphi(T) \in \sem{P}\), because \(\tuple{y},\tuple{y}' \not\in P\).

    \item[AG] First, we show the implication \(\rightarrow\): Let \(\omega \in \sem{[\daesys{\tuple{x}}{F}]P}\).
      Assume that \(\omega \in \sem{F}\). 
      Otherwise, there is no flow that satisfies the differential constraint \(F\). 
      Also, fix arbitrary \(T \geq 0, \epsilon > 0\) and a flow \(\varphi \in \analyticFlows{\tuple{x},\tuple{y}}{-\epsilon, T + \epsilon}\) with \(\varphi([0,T]) \subseteq \sem{F \land G(\tuple{y})}\) and \(\varphi(0) = \omega\). 
      We need to show \(\varphi(T) \in \sem{P}\).
      In particular, we have \(\varphi \subseteq \sem{F}\). 
      By assumption, this implies \(\varphi(T) \in \sem{P}\).

      Conversely, we show the implication \(\leftarrow\):
      Fix an arbitrary initial state \(\omega\) with \(\omega \in
      \sem{[\daesys{\tuple{x}}{F}]G(h(\tuple{x},\tuple{x}'))}\) and \(\omega \in \sem{\forall \tuple{y} \forall \tuple{y}'[\daesys{\tuple{x},\tuple{y}}{F\land G(\tuple{y})}]P}\).
      Fix \(T \geq 0, \epsilon > 0\) and an arbitrary flow \(\varphi \in \analyticFlowsInterval\) with \(\varphi([0,T]) \subseteq \sem{F}\) and \(\omega = \varphi(0)\). 
      We need to show that \(\varphi(T) \in \sem{F}\).
      By assumption, we have \(\varphi([0,T]) \subseteq \sem{G(h(\tuple{x},\tuple{x}'))}\).
      Now, we construct a mapping \(\varphi_y :(-\epsilon, T + \epsilon) \rightarrow \states\) with \(\varphi_y(0) = \omega\) and satisfies \(\varphi_y \subseteq \sem{F \land G(\tuple{y})}\), for any fresh variable \(\tuple{y},\tuple{y}' \not \in F, G, h(\cdot,\cdot)\).
      The key idea is that \(h\) implicitly defines a possible \(\tuple{y}\) as a function of \(\tuple{x}\) and \(\tuple{x}'\).
      Note that since, the term language of \dAL consists only of
      polynomial terms, \(h : \R^{2n} \rightarrow \R\) is a globally
      defined, real analytic function.
      This allows us to define \(\varphi_y\) as follows:
      \begin{align*}
        \varphi_y(t)(\tuple{y}) &= h(\varphi(t)(\tuple{x}),\varphi(t)(\tuple{x}')), \\
        \varphi_y(t)(\tuple{y}') &=
        \frac{\partial h}{\partial
        x}(\varphi(t)(\tuple{x}),\varphi(t)(\tuple{x}'))\varphi(t)(\tuple{x}') +
        \frac{\partial h}{\partial
        x'}(\varphi(t)(\tuple{x}),\varphi(t)(\tuple{x}'))\frac{d \varphi(t)(\tuple{x}')}{dt}, \\
        \varphi_y(t)(z) &= \varphi(t)(z) \text{ on } z \in (\tuple{y}
        \cup \tuple{y}')^\complement.
      \end{align*}
      
        By \rref{lem:analytic}, we see that \(\varphi_y(\cdot)(\tuple{x})\) is analytic on \((-\epsilon, T + \epsilon)\), because
        \(\varphi(\cdot)(\tuple{x})\) is analytic on \((-\epsilon, T + \epsilon)\).  Also, \(\varphi_y\) is differentially consistent in \(\tuple{y},\tuple{y}'\) by:
      \begin{align*}
        \frac{d \varphi_y(t)}{dt}(\tuple{y}) &= \frac{\partial
        h}{\partial x}(\varphi(t)(\tuple{x}),\varphi(t)(\tuple{x}'))\frac{d
        \varphi(t)}{dt}(\tuple{x}) +
        \frac{\partial h}{\partial
        x'}(\varphi(t)(\tuple{x}),\varphi(t)(\tuple{x}'))\frac{d \varphi(t)(\tuple{x}')}{dt}\\
        &= \frac{\partial h}{\partial
        x}(\varphi(t)(\tuple{x}),\varphi(t)(\tuple{x}'))\varphi(t)(\tuple{x}') +
        \frac{\partial h}{\partial
        x'}(\varphi(t)(\tuple{x}),\varphi(t)(\tuple{x}'))\frac{d \varphi(t)(\tuple{x}')}{dt}\\
        &= \varphi_y(t)(\tuple{y}').
      \end{align*}
      This implies \(\varphi_y \in \analyticFlowsInterval[y]\).
      Then by construction we have \(\varphi_y([0,T]) \subseteq \sem{F}\), by \rref{lem:coincidence} \(\tuple{y},\tuple{y}' \not\in F\), because \(\varphi\) and \(\varphi_y\) only differ on \(\tuple{y}\cup \tuple{y}'\).
      Also, we have \(\varphi_y([0,T])\subseteq \sem{G(\tuple{y})}\). 
      Therefore, we have \(\varphi_y(T) \in \sem{P}\) by \(\omega \in \sem{\forall \tuple{y}\forall \tuple{y}'[\daesys{\tuple{x},\tuple{y}}{F\land G(\tuple{y})}]P}\).
      By construction of \(\varphi_y\), we have \(\varphi(T) =
      \varphi_y(T)\) on \((\tuple{y}\cup \tuple{y}')^\complement\) and
      therefore \(\varphi(T) \in \sem{P}\), because
      \(\tuple{y},\tuple{y}' \not\in P\) and \rref{lem:coincidence}.
      \item[DI] We assume, without loss of generality, that \(P
        \equiv e(\tuple{x}) \geq 0\), where \((P)' \equiv (e(\tuple{x}))' \geq 0\).
        All other cases follow analogously.
        Fix an arbitrary state \(\omega \in \sem{[?F]P}\) and
        \(\omega \in \sem{F \rightarrow [\daesys{\tuple{x}}{F}](P)'}\).
        If \(\omega \not\in \sem{F}\), then there are no \(T \geq 0, \epsilon > 0\) and \(\varphi \in
        \analyticFlowsInterval\) that satisfy \(F\) and there is
        nothing to show.
        Otherwise, we have \(\omega \in
      \sem{[\daesys{\tuple{x}}{F}](P)'}\) and \(\omega \in \sem{P}\) by \(\omega \in
        \sem{[?F]P}\).
        Now, let \(T \geq 0, \epsilon > 0\) and \(\varphi \in \analyticFlowsInterval\)
      be a flow with \(\varphi([0,T]) \subseteq \sem{F}\) and \(\varphi(0) =
        \omega\). The case \(T = 0\) follows immediately, because
        \(\varphi(0) = \omega \in \sem{P}\).
        Now, assume
      \(T > 0\) this implies \(0 \leq \varphi(\xi)\sem{(e(\tuple{x}))'} =
      \frac{d \varphi(t)\sem{e(\tuple{x})}}{dt}(\tau)\) for all \(\tau \in
      [0,T]\), by \rref{lem:diff}.
      By the mean value theorem we have \(\varphi(T)\sem{e(\tuple{x})} - \varphi(0)\sem{e(\tuple{x})} = T \frac{d\varphi(t)\sem{e(\tuple{x})}}{dt}(\xi)\) for some \(\xi \in (0,T)\),
        since \(\varphi(\cdot)\sem{e(\tuple{x})}\) is analytic.
        Therefore, we have \(\varphi(T)\sem{e(\tuple{x})} - \varphi(0)\sem{e(\tuple{x})} \geq 0\).
        By \(\varphi(0) \in \sem{e(\tuple{x}) \geq 0}\), this implies \(\varphi(T)
        \in \sem{e(\tuple{x}) \geq 0}\).
      \item[GS] Fix an arbitrary state \(\omega\) with \(\omega \in \sem{[\daesys{\tuple{x}}{F \lor G}](\consis{F}{G} \land \consis{G}{F})}\) and \(\omega \in
        \sem{F \lor G \rightarrow [\{\daesys{\tuple{x}}{F} \cup
      \daesys{\tuple{x}}{G}\}^*]P}\).
      We need to show \(\omega \in \sem{[\daesys{\tuple{x}}{F \lor G}]P}\): Assume \(\omega \in \sem{F
      \lor G}\).
      Otherwise, there is no flow satisfying the
      differential constraint \(F \lor G\) and there is nothing to prove.
      Now, fix \(T \geq 0, \epsilon > 0\) and \(\varphi \in \analyticFlowsInterval\) such that \(\varphi([0,T]) \subseteq \sem{F \lor G}\) and \(\varphi(0) = \omega\).
      Without loss of generality assume \(\omega \in \sem{F}\); the case \(\omega \in \sem{G}\) is analogous.
      Suppose there exists a monotonically increasing sequence \((\tau_{k})_{k \in \N_0}\) starting in zero with \(\varphi(0) = \varphi(\tau_0) = \omega\), \(\varphi(\tau_k) \in \sem{F}\) and \(\varphi(\tau_{k+1}) \in \sem{G}\) for \(k \geq 1\).
      Because we have \(\varphi([0,T]) \subseteq \sem{\consis{F}{G}}\) and \(\varphi([0,T]) \subseteq \sem{\consis{G}{F}}\) by assumption, using \rref{lem:seq}, there is a \(k_0 \in
      \N\) such that \(\varphi(\tau) \in \sem{F}\) for all \(\tau_{k_0} \leq \tau \leq T\) or  \(\varphi(\tau) \in \sem{G}\) for all \(\tau_{k_0} \leq \tau \leq T\).
      By \rref{lem:mode_consistency}, there exists a sequence \((t_k)_{k
      =0}^{k_0-1}\) with \(\tau_k \leq t_k \leq \tau_{k+1}\) such
      that \(\varphi([\tau_k, t_k]) \subseteq \sem{F}\) and
      \(\varphi([t_k, \tau_{k+1}]) \subseteq \sem{G}\) for all \(k < k_0\)
      and either \(\varphi([\tau_{k_0},T]) \subseteq \sem{F}\) or
      \(\varphi([\tau_{k_0},T]) \subseteq \sem{G}\).
      Next, we transform the loop assumption to apply it to the flow \(\varphi\). By \rref{def:pro}, we have
      \begin{align*}
        \omega &\in \sem{[\{\daesys{x}{F} \cup
        \daesys{\tuple{x}}{G}\}^*]P} = \bigcup_{k\in\N} \sem{[\{\daesys{\tuple{x}}{F} \cup
        \daesys{\tuple{x}}{G}\}^k]P}.
      \end{align*}
      In particular, this implies \(\omega \in \sem{[\{\daesys{\tuple{x}}{F} \cup
      \daesys{\tuple{x}}{G}\}^{k_0 + 1}]P}\).
      Unrolling the loop one step, we have \(\omega \in \sem{[\daesys{\tuple{x}}{F} \cup
      \daesys{\tuple{x}}{G}][\{\daesys{\tuple{x}}{F} \cup
      \daesys{\tuple{x}}{G}\}^{k_0}]P}\).
      By \rref{def:pro}, we have \(\omega \in \sem{[\daesys{\tuple{x}}{F}][\{\daesys{x}{F} \cup
      \daesys{\tuple{x}}{G}\}^{k_0}]P}\).
      We can inductively repeat this argument to obtain \(\omega \in \sem{[\daesys{\tuple{x}}{F}][\daesys{\tuple{x}}{G}]\dots[\daesys{\tuple{x}}{F}]P}\) and \(\omega \in \sem{[\daesys{\tuple{x}}{F}][\daesys{\tuple{x}}{G}]\dots[\daesys{\tuple{x}}{G}]P}\).
      By the construction of \(\varphi\), this implies \(\varphi([t_k, \tau_{k+1}])\subseteq \sem{P}\) for all \(k = 1,\dots, k_0-1\) or \(\varphi([\tau_{k_0}, T]) \subseteq \sem{P}\).
      This implies
      \[
        \varphi([0,T]) = \varphi([\tau_{k_0}, T]) \cup \bigcup_{k=0}^{k_0}\varphi([t_k, \tau_{k+1}])\subseteq \sem{P}.
      \]
      \qed
  \end{itemize}
\end{proof}

\subsection{Derived Rules}

\begin{proof}
  We derive the rules from the base axioms and proof rules:
  \begin{enumerate}
    \item[AR] We start the proof by using an algebraic ghost on the box in the consequent:
      \begin{prooftree}
              \AX$\fCenter [\daesys{\tuple{x}}{G}]P  \vdash  \forall \tuple{y} \forall \tuple{y}'[\daesys{\tuple{x},\tuple{y}}{F\land G}]P$
              \AX$\fCenter \forall \tuple{x} \forall \tuple{x}'(F \rightarrow G) \vdash
              [\daesys{\tuple{x}}{F}]G$
              \def\defaultHypSeparation{\hskip 1pt}
        \LeftLabel{WL,AR}
              \BI$\fCenter \forall \tuple{x} \forall \tuple{x}'(F \rightarrow
            G),[\daesys{\tuple{x}}{G}]P
              \vdash  [\daesys{\tuple{x}}{F}]P$
        \LeftLabel{\impliesright}
              \UI$\fCenter \vdash \forall \tuple{x} \forall \tuple{x}' (F \rightarrow G) \rightarrow
              ([\daesys{\tuple{x}}{G}]P \rightarrow [\daesys{\tuple{x}}{F}]P)$
      \end{prooftree}
      Note that we have the side condition \(\tuple{y},\tuple{y}' \not \in F,G\). This can be enforced by using fresh variables \(\tuple{y},\tuple{y}'\).
      We use the derived axiom \[\gray{\text{DW}\rightarrow} \quad [\daesys{\tuple{x}}{F}]G \leftrightarrow [\daesys{\tuple{x}}{F}](F \rightarrow G),\] which can be derived from K, DW and G.
      Using this axiom and the G rule, we can use
      the algebraic condition to close the right goal:
      \begin{prooftree}
              \AX$\fCenter *$
        \LeftLabel{\(\R\)}
              \UI$\forall \tuple{x} \forall \tuple{x}'(F \rightarrow G)\fCenter\,\vdash F \rightarrow G$
        \LeftLabel{\g}
              \UI$\forall \tuple{x} \forall \tuple{x}'(F \rightarrow G)\fCenter\,\vdash [\daesys{\tuple{x}}{F}](F \rightarrow G)$
        \LeftLabel{\dwimplies}
              \UI$\forall \tuple{x} \forall \tuple{x}'(F \rightarrow G)\fCenter\,\vdash [\daesys{\tuple{x}}{F}]G$
      \end{prooftree}
      The application of the G rule is sound, because we have \[\BV{\daesys{\tuple{x}}{F}} \cap \FV{\forall x \forall x'(F \rightarrow G)} = \varnothing.\]
      Finally, we use the C and DR axioms, to remove \(F\)
      from the box in the goal to finish the proof:
      \begin{prooftree}
        \AX$\fCenter \quad \quad \quad \quad \;\;*$
        \LeftLabel{id}
        \UI$\fCenter[\daesys{\tuple{x}}{G}]P  \vdash  [\daesys{\tuple{x}}{G}]P$
        \LeftLabel{DR}
        \UI$\fCenter[\daesys{\tuple{x}}{G}]P  \vdash  \forall \tuple{y} \forall \tuple{y}'[\daesys{\tuple{x},\tuple{y}}{G\land F}]P$
        \LeftLabel{C}
        \UI$\fCenter[\daesys{\tuple{x}}{G}]P  \vdash  \forall \tuple{y} \forall \tuple{y}'[\daesys{\tuple{x},\tuple{y}}{F\land G}]P$
      \end{prooftree}
      \item[DC] We derive the differential cut axiom. First, use the implies right rule:
      \begin{prooftree}
        \AX$\fCenter [\daesys{\tuple{x}}{F}]G \vdash [\daesys{\tuple{x}}{F \land G}]P \leftrightarrow[\daesys{\tuple{x}}{F}]P$
        \LeftLabel{\(\rightarrow\)R}
        \UI$\fCenter \vdash [\daesys{\tuple{x}}{F}]G \rightarrow \left([\daesys{\tuple{x}}{F \land G}]P \leftrightarrow[\daesys{\tuple{x}}{F}]P\right)$
      \end{prooftree}
      We split the equivalence into two goals:
      \begin{prooftree}
        \AX$\fCenter [\daesys{\tuple{x}}{F}]G, [\daesys{\tuple{x}}{F}]P \vdash [\daesys{\tuple{x}}{F \land G}]P$
        \AX$\fCenter [\daesys{\tuple{x}}{F}]G, [\daesys{\tuple{x}}{F \land G}]P \vdash [\daesys{\tuple{x}}{F}]P$
        \LeftLabel{\(\leftrightarrow\)}
        \BI$\fCenter [\daesys{\tuple{x}}{F}]G \vdash [\daesys{\tuple{x}}{F \land G}]P \leftrightarrow[\daesys{\tuple{x}}{F}]P$
      \end{prooftree}
      For the first goal, we use the AG axiom to augment the system with the ghost variable \(y\) using the formula \(\top \equiv x^2 \geq 0\), which holds in all states. This augmentation does not change the behavior of the systems but allows us to add the unrestricted ghost variable \(y\).
      \begin{prooftree}
        \AX$\fCenter  \; *$
        \LeftLabel{\(\R\)}
        \UI$\fCenter  \; \vdash \top$
        \LeftLabel{G}
        \UI$ \fCenter \; \vdash [\daesys{\tuple{x}}{F \land G}]\top$
        \AX$[\daesys{\tuple{x}}{F}]G, [\daesys{\tuple{x}}{F}]P \fCenter \;\vdash \forall \tuple{y} \forall \tuple{y}'[\daesys{\tuple{x},\tuple{y}}{F \land G \land \top}]P$
        \LeftLabel{AG}
        \BI$[\daesys{\tuple{x}}{F}]G, [\daesys{\tuple{x}}{F}]P \fCenter \;\vdash [\daesys{\tuple{x}}{F \land G}]P$
      \end{prooftree}
      Now, we use the AR to get rid of the additional formula and close the remaining goal:
      \begin{prooftree}
        \AX$\fCenter *$
        \LeftLabel{id}
        \UI$[\daesys{\tuple{x}}{F}]P \fCenter \;\vdash [\daesys{\tuple{x}}{F}]P$
        \LeftLabel{DR}
        \UI$[\daesys{\tuple{x}}{F}]P \fCenter \;\vdash \forall \tuple{y} \forall \tuple{y}'[\daesys{\tuple{x},\tuple{y}}{F \land G}]P$
        \LeftLabel{AR}
        \UI$[\daesys{\tuple{x}}{F}]P \fCenter \;\vdash \forall \tuple{y} \forall \tuple{y}'[\daesys{\tuple{x},\tuple{y}}{F \land G \land \top}]P$
      \end{prooftree}
      The second goal can be derived analogously.
      \begin{prooftree}
        \AX$\fCenter *$
        \LeftLabel{id}
        \UI$[\daesys{\tuple{x}}{F}]G, [\daesys{\tuple{x}}{F \land G}]P \; \fCenter \vdash [\daesys{\tuple{x}}{F \land G}]P$
        \LeftLabel{R,AR}
        \UI$[\daesys{\tuple{x}}{F}]G, [\daesys{\tuple{x}}{F \land G}]P \; \fCenter \vdash \forall \tuple{y} \forall \tuple{y}'[\daesys{\tuple{x},\tuple{y}}{F \land G}]P$
        \LeftLabel{AG}
        \UI$[\daesys{\tuple{x}}{F}]G, [\daesys{\tuple{x}}{F \land G}]P \; \fCenter \vdash [\daesys{\tuple{x}}{F}]P$
      \end{prooftree}
      \item[\(\land\)DE] Splitting the equivalence, we have:
      \begin{prooftree}
        \AX$\fCenter [\daesys{\tuple{x}}{F \land e =
        0}]P\vdash  [\daesys{\tuple{x}}{F \land e = 0 \land (e)' = 0}]P$
        \noLine
        \UI$\fCenter [\daesys{\tuple{x}}{F \land e = 0 \land (e)' = 0}]P\vdash [\daesys{\tuple{x}}{F \land e =
        0}]P$
        \LeftLabel{\(\leftrightarrow\)}
        \UI$\fCenter \vdash [\daesys{\tuple{x}}{F \land e =
        0}]P \leftrightarrow [\daesys{\tuple{x}}{F \land e = 0 \land (e)' = 0}]P$
      \end{prooftree}
      The first proof goal follows from R. We transform the second goal using DC:
      \begin{prooftree}
        \AX$\fCenter [\daesys{\tuple{x}}{F \land e = 0 \land (e)' = 0}]P\vdash [\daesys{\tuple{x}}{F \land e =
        0 \land (e)' = 0}]P$
        \noLine
        \UI$\fCenter [\daesys{\tuple{x}}{F \land e = 0 \land (e)' = 0}]P\vdash [\daesys{\tuple{x}}{F \land e =
        0}](e)'= 0$
        \LeftLabel{DC}
        \UI$\fCenter [\daesys{\tuple{x}}{F \land e = 0 \land (e)' = 0}]P\vdash [\daesys{\tuple{x}}{F \land e =
        0}]P$
      \end{prooftree}
      Finally, the first goal closes by id and the second goal by AR and DE.
  \end{enumerate}
  \qed
\end{proof}

\section{Extended Example Proofs} \label{sec:examples}
In this section, we show the detailed proof to close the goal involving the mode consistency formula. Using the \(G\) rule we get rid of the box modality we obtain
\begin{prooftree}
  \AX$\fCenter \vdash \consisxy{G(x,y)}{F(x,y)}$
  \AX$\fCenter \vdash \consisxy{F(x,y)}{G(x,y)}$
  \LeftLabel{\(\land\)R}
  \BI$\fCenter \vdash \consisxy{F(x,y)}{G(x,y)} \land \consisxy{G(x,y)}{F(x,y)}$
  \LeftLabel{G}
  \UI$\fCenter \vdash [\{\daesys{x}{F(x,y) \lor G(x,y)}\}](\consisxy{F(x,y)}{G(x,y)} \land \consisxy{G(x,y)}{F(x,y)})$
\end{prooftree}
This leaves two open goal to prove. We proceed by computing the entry and exit formulas of $F(x,y)$ and $G(x,y)$.
The exit formula of $F(x,y)$ is given by:
\begin{align*}
    &\exit{\daesys{x}{F(x,y)}} \equiv \\
    &\quad (x' = 0 \land x^2=1 \land y = 0) \land \progress{\daesys{x,y}{x' \neq 0 \lor x^2 \neq 1 \lor y \neq 0}}\\
    &\quad\quad\lor (x' \neq 0 \lor x^2 \neq 1 \lor y \neq 0) \land \progress{\daesys{x,y}{x' = 0 \land x^2=1 \land y = 0}}^- \\
\end{align*}
This implies \(\exit{\daesys{x}{F(x,y)}} \equiv x' = 0 \land x^2=1 \land y = 0\).
Its entry formula is given by:
\begin{align*}
  &\enter{\daesys{x}{F(x,y)}} \equiv \\
  &\quad \progress{\daesys{x,y}{x' = 0 \land x^2=1 \land y = 0}} \land (x' \neq 0 \lor x^2 \neq 1 \lor y \neq 0)\\
  &\quad\quad\lor \progress{\daesys{x,y}{x' \neq 0 \lor x^2 \neq 1 \lor y \neq 0}}^- \land (x' = 0 \land x^2=1 \land y = 0) \\
\end{align*}
This also simplifies to \(\enter{\daesys{x}{F(x,y)}} \equiv x' = 0 \land x^2=1 \land y = 0\).

The exit formula of $G(x,y)$ is given by:
\begin{align*}
  &\exit{\daesys{x,y}{G(x,y)}} \equiv \\
  &\progress{\daesys{x,y}{x' \neq y \lor x^2 + y^2 \neq 1  \lor y'\neq -x}} \land (x' = y \land x^2 + y^2 = 1  \land y'=-x)\\
  &\quad \lor \progress{\daesys{x,y}{x' = -y \land x^2 + y^2 = 1  \land y'=x}}^- \land x' \neq y \lor x^2 + y^2 \neq 1  \lor y'\neq -x\\
\end{align*}
This simplifies to \(\exit{\daesys{x,y}{G(x,y)}} \equiv x' = y \land x^2 + y^2 = 1  \land y'=-x\).
Finally, the entry formula of G(x,y) is given by:
\begin{align*}
  &\enter{\daesys{x,y}{G(x,y)}} \equiv \\
  &\progress{\daesys{x,y}{x' = y \land x^2 + y^2 = 1  \land y'=-x}} \land x' \neq y \lor x^2 + y^2 \neq 1  \lor y'\neq -x\\
  &\quad\lor \progress{\daesys{x,y}{x' \neq -y \lor x^2 + y^2 \neq 1  \lor y'\neq x}}^- \land (x' = y \land x^2 + y^2 = 1  \land y'=-x)\\
\end{align*}
Which also simplifies to \(\enter{\daesys{x,y}{G(x,y)}} \equiv x' = y \land x^2 + y^2 = 1  \land y'=-x\).
Now, using the computed formulas, we can easily verify using real arithmetic reasoning that both 
\begin{align*}
  \exit{\daesys{x}{F(x,y)}} \land  \enter{\daesys{x,y}{G(x,y)}} &\vdash F(x,y) \leftrightarrow G(x,y),\\
  \exit{\daesys{x}{G(x,y)}} \land  \enter{\daesys{x,y}{F(x,y)}} &\vdash F(x,y) \leftrightarrow G(x,y).
\end{align*}
This finishes the proof.
\end{document}